\documentclass[10pt,twocolumn,english]{IEEEtran}
\usepackage{amsmath,amsfonts,amssymb,algorithm,algorithmic,url,comment}
\usepackage[dvips]{color}
\usepackage{graphicx,tabulary,multirow}
\usepackage[amsmath,thmmarks,amsthm]{ntheorem}
\usepackage{epsfig,subfig,cite}

\newcommand{\ie}{{\em i.e. }}
\newcommand{\eg}{{\em e.g. }}
\newtheorem{theorem}{Theorem}

\begin{document}

\title{Sub-Nyquist Radar via Doppler Focusing}

\author{Omer Bar-Ilan and Yonina C. Eldar,~\IEEEmembership{Fellow,~IEEE}
\thanks{The authors are with the Department of Electrical Engineering, Technion--Israel Institute of Technology, Haifa, Israel 32000 (phone: +972-4-8294798, +972-4-8294682, fax: +972-4-8295757, e-mail: omerba@tx.technion.ac.il, yonina@ee.technion.ac.il). }}

\maketitle

\begin{abstract}
We investigate the problem of a monostatic pulse-Doppler radar transceiver trying to detect targets, sparsely populated in the radar's unambiguous time-frequency region. Several past works employ compressed sensing (CS) algorithms to this type of problem, but either do not address sample rate reduction, impose constraints on the radar transmitter, propose CS recovery methods with prohibitive dictionary size, or perform poorly in noisy conditions. Here we describe a sub-Nyquist sampling and recovery approach called Doppler focusing which addresses all of these problems: it performs low rate sampling and digital processing, imposes no restrictions on the transmitter, and uses a CS dictionary with size which does not increase with increasing number of pulses $P$. Furthermore, in the presence of noise, Doppler focusing enjoys a signal-to-noise ratio (SNR) improvement which scales linearly with $P$, obtaining good detection performance even at SNR as low as -25dB. The recovery is based on the Xampling framework, which allows reducing the number of samples needed to accurately represent the signal, directly in the analog-to-digital conversion process. After sampling, the entire digital recovery process is performed on the low rate samples without having to return to the Nyquist rate. Finally, our approach can be implemented in hardware using a previously suggested Xampling radar prototype.
\end{abstract}


\section{Introduction}
\label{sec:Introduction}

We consider target detection and parameter estimation in a pulse-Doppler radar system, using sub-Nyquist sampling rates. The radar is a single transceiver, monostatic, narrow-band system. Targets are non-fluctuating point targets, sparsely populated in the radar's unambiguous time-frequency region: delays up to the Pulse Repetition Interval (PRI), and Doppler frequencies up to its reciprocal, the Pulse Repetition Frequency (PRF). We propose a recovery method which can detect and estimate targets' time delay and Doppler frequency, using a linear, non-adaptive sampling technique at a rate significantly lower than the radar signal's Nyquist frequency, assuming the number of targets $L$ is small.

Current state-of-the-art radar systems sample at the signal's Nyquist rate, which can be hundreds of MHz and even up to several GHz. Systems exploiting sub-Nyquist sampling rates benefit from a lower rate analog-to-digital conversion (ADC), which requires less computational power. Moreover, sampling at the Nyquist rate may not always be feasible due to high power consumption, heat dissipation, cost, or other practical considerations. Finally, offline radar systems which record samples for subsequent processing can gain substantial storage capacity reduction if they were to sample at sub-Nyquist rates.

The goal of this work is to present some steps in order to break the link between radar signal bandwidth and sampling rate. We mainly focus on the simpler setting of Swerling-0 point targets with additive white Gaussian noise (AWGN) and ideal hardware, in order to bring forth what we believe are useful ideas and algorithms for sample rate reduction. We also briefly treat target dynamic range and clutter. However, a full analysis of these important issues is left to future work. The sub-Nyquist Xampling (``compressed sampling") \cite{Xampling} method we use is an ADC which performs analog prefiltering of the signal before taking point-wise samples. These compressed samples (``Xamples") contain the information needed to recover the desired signal parameters using compressed sensing (CS) algorithms. This work expands the work in \cite{EliGal}, adding Doppler to the target model and proposing a new method to estimate it. The same sampling technique and hardware that were used in \cite{EliGal, multichannel} will also work here, while the digital processing we suggest is adapted to moving targets, and low signal-to-noise ratio (SNR).

Past works employ CS algorithms to this type of problem, but do not address sample rate reduction and continue sampling at the Nyquist rate \cite{HighResRadar, BlockSparse}. Other works combine radar and CS in order to reduce the receiver's sampling rate, but in doing so impose constraints on the radar transmitter and do not treat noise \cite{CompressiveRadarImaging}, or do not handle noise well \cite{KfirWaheed}. The work in \cite{KfirWaheed} first estimates target delays and then uses these recovered delays to estimate Doppler frequencies and amplitudes, and is referred to as ``two-stage recovery" in subsequent sections. Another line of work proposes single stage CS recovery methods with dictionary size proportional to the product of delay and Doppler grid sizes, making them infeasible for many realistic scenarios \cite{HighResRadar, AdaptiveCSRadar}.

Our approach is based on the observation that the received radar signal can be modeled with $3L$ degrees of freedom (DOF): a delay, Doppler frequency and amplitude for each of the $L$ targets. Signals which can be described with a fixed number of DOF per unit of time are known as Finite Rate of Innovation (FRI) \cite{VetterliFRI} signals. The proposed recovery process is actually a recovery of these DOF from low rate samples. The concept of FRI together with the Xampling methodology enables sub-Nyquist rates using practical hardware \cite{Xampling}.

At the crux of our proposed recovery method is a coherent superposition of time shifted and modulated pulses, the Doppler focusing function $\Phi(t;\nu)$. For any Doppler frequency $\nu$, this function combines the received signals from different pulses so targets with appropriate Doppler frequencies come together in phase. For each sought after $\nu$, $\Phi(t;\nu)$ is processed as a simple one-dimensional CS problem and the appropriate time delays are recovered. The gain from this method is both in terms of SNR and Doppler resolution. For $P$ pulses adding coherently, we obtain a factor $P$ SNR improvement over white noise (which adds incoherently, \ie in power), as will be proved in Section~\ref{sec:DelayDopplerRecovery}. Such an SNR improvement is optimal, it is the improvement obtained by a matched filter (MF), and we show this in Section~\ref{sec:ClassicProcessing}. We also analyze the minimal number of samples required for perfect recovery without noise, and show that Doppler focusing attains this lower bound in terms of number of samples per pulse. In addition, denoting the PRI as $\tau$, we demonstrate in Section~\ref{sec:DopplerFocusing} that the width of the Doppler focus for each $\Phi(t;\nu)$ is $2\pi /P \tau$, meaning that delays of targets separated in Doppler by more than $2\pi/P \tau$ create almost no interference with each other.

The idea of Doppler focusing comes from a similar function in the context of ultrasound beamforming used in \cite{CompressedBeamformingUltrasound}. There, in a method named ``Dynamic focusing", the signal returned to a set of linearly aligned transceivers is focused in a manner similar to how we focus pulses, and the Doppler frequency $\nu$ is replaced by spatial direction $\theta$. In both cases, the advantages of focusing are not lost with sub-Nyquist processing since it can be performed on the low rate Xamples.

Simulations provided in Section~\ref{sec:SimulationResults} show that when sampling at one tenth the Nyquist rate, our Doppler focusing recovery method outperforms both classic MF recovery, described in Section~\ref{sec:ClassicProcessing}, and two-stage CS recovery, described in Section~\ref{sec:PreviousApproaches}. When the SNR reaches -25dB, our method achieves the performance of classic recovery operating at the full Nyquist rate.

The main merits of our proposed method are as follows:
\begin{enumerate}
\item \textbf{Low rate ADC and DSP} -- using Xampling and the proposed recovery method, we are able to acquire the sub-Nyquist samples containing information needed for target recovery, and then digitally recover the unknown target parameters using low rate processing, without returning to the higher Nyquist rate.
\item \textbf{Transmitter compatibility} -- our recovery method does not impose any restrictions on the transmitted signal, provided it meets the assumptions stated in Section~\ref{sec:RadarModel}.
\item \textbf{Scaling with problem size} -- many CS delay-Doppler estimation methods depend upon constructing a CS dictionary with a column for each delay-Doppler hypothesis. For even moderate size problems, this requires a huge dictionary, making them infeasible for many systems, especially real-time ones. Our Doppler focusing based method avoids this problem by separating the Doppler from delay recovery, making each CS delay recovery indifferent to the underlying Doppler.
\item \textbf{Robustness to noise and clutter} -- The SNR achieved using Doppler focusing scales linearly with the number of received pulses $P$, as does an optimal MF, providing good performance in AWGN. Regarding clutter, Doppler focusing includes inherent isolation between signals with different Doppler frequencies, so unless target and clutter have very similar Doppler frequency, target detection is unhindered.
\end{enumerate}

The remainder of this paper is organized as follows. In Section~\ref{sec:RadarModel} we describe the radar model and the assumptions used for simplification. Section~\ref{sec:ClassicProcessing} reviews classic MF processing. We explain the Doppler focusing concept in Section~\ref{sec:DopplerFocusing} and sub-Nyquist delay recovery in Section~\ref{sec:DelayRecovery}. The delay-Doppler recovery method using Doppler focusing is described in Section~\ref{sec:DelayDopplerRecovery}, along with an analysis of noiseless recovery and some practical considerations. We provide a comparison with other CS recovery methods in Section~\ref{sec:PreviousApproaches}. Numerical results are presented in Section~\ref{sec:SimulationResults}.

We denote vectors by boldface lower case letters, \eg $\mathbf{c}$, and matrices by boldface capital letters, \eg $\mathbf{A}$. We say a vector $\mathbf{c}$ is $L$-sparse if $\| \mathbf{c} \|_0 \leq L$, \ie at most $L$ of its elements are nonzero. The $n$th element of a vector is written as $\mathbf{c}_n$, and the $ij$th element of a matrix is denoted by $\mathbf{A}_{ij}$. Non-boldface variables represent scalars or functions, where continuous functions are denoted with round parentheses, \eg $x(t)$ and discrete functions with square parentheses, \eg $h[n]$. The cardinality of a set $\kappa$ is denoted $|\kappa|$.

\section{Radar Model}
\label{sec:RadarModel}

We consider a radar transceiver that transmits a pulse train
\begin{equation}
\label{eq:TxSignal}
x_T(t)=\sum_{p=0}^{P-1}h(t-p\tau), \; \quad 0\leq t \leq P\tau
\end{equation}
consisting of $P$ equally spaced pulses $h(t)$. The pulse-to-pulse delay $\tau$ is referred to as the PRI, and its reciprocal $1/\tau$ is the PRF. The entire span of the signal in \eqref{eq:TxSignal} is called the coherent processing interval (CPI). The pulse $h(t)$ is a known time-limited baseband function with continuous-time Fourier transform (CTFT) $H(\omega)=\int_{-\infty}^{\infty} h(t) e^{-j \omega t} dt$. We assume that $H(\omega)$ has negligible energy at frequencies beyond $B_h/2$ and we refer to $B_h$ as the bandwidth of $h(t)$. The target scene is composed of $L$ non-fluctuating point targets (Swerling-0 model, see \cite{RadarHandbook}), where we assume that $L$ is known, although this assumption can easily be relaxed. The pulses reflect off the $L$ targets and propagate back to the transceiver. Each target $\ell$ is defined by three parameters: a time delay $\tau_\ell$, proportional to the target's distance from the radar; a Doppler radial frequency $\nu_\ell$, proportional to the target-radar closing velocity; and a complex amplitude $\alpha_\ell$, proportional to the target's radar cross section (RCS), dispersion attenuation and all other propagation factors. We limit ourselves to defining targets in the radar's radial coordinate system.

Throughout, we make the following assumptions on the targets' location and motion, which leads to a simplified expression for the received signal. For this we define a few auxiliary parameters that were previously unneeded. Denote the radar's carrier frequency as $f_c$, the pulse time support as $T_p$, and the speed of light as $c$. We use the time-distance equivalence $r_\ell=c\tau_\ell/2$ and the non-relativistic Doppler radial frequency-velocity equivalence $\dot{r}_\ell=\nu_\ell c/4 \pi f_c$.
\begin{enumerate}
\item[\textbf{A1.}] ``Far targets" –- target-radar distance is large compared to the distance change during the CPI which allows for constant $\alpha_\ell$.
\begin{equation}
\dot{r}_\ell P \tau \ll r_\ell \Rightarrow \nu_\ell \ll 2 \pi f_c \tau_\ell / P \tau.
\end{equation}
\item[\textbf{A2.}] ``Slow targets" -– small target velocity allows for constant $\tau_\ell$ during the CPI and constant Doppler phase during pulse time $T_p$.
\begin{equation}
2 \dot{r}_\ell B_h / c \ll 1/P \tau \Rightarrow \nu_\ell \ll 2 \pi f_c / P \tau B_h.
\end{equation}
This inequality is termed the ``narrowband assumption" in radar nomenclature, since it gives an upper limit on the radar signal's bandwidth.
\item[\textbf{A3.}] ``Small acceleration" –- target velocity remains approximately constant during the CPI allowing for constant $\nu_\ell$.
\begin{equation}
\ddot{r}_\ell P \tau \ll c / 2 f_c P \tau \Rightarrow \ddot{r}_\ell \ll c / 2 f_c (P \tau)^2.
\end{equation}
\end{enumerate}

Although these assumptions may seem hard to comply with, they all rely on slow ``enough" relative motion between the radar and its targets. Radar systems tracking people, ground vehicles and sea vessels usually comply quite easily. For example, consider a $P$=100 pulse radar system with PRI $\tau$=100$\mu$sec, pulse width $T_p$=1$\mu$sec, bandwidth $B_h$=30MHz and carrier frequency $f_c$=3GHz, tracking cars traveling up to 120km/hour. The maximal distance change over the CPI is approximately 0.33m, so if the targets' minimal distance from the radar is a few meters, then $\textbf{A1.}$ is satisfied. As for $\textbf{A2.}$, the maximal Doppler frequency is approximately 667Hz, which is much smaller than both $f_c/P\tau B_h$=10KHz and $1/T_p$=1MHz. An extreme acceleration of 10m/sec$^2$ would cause a velocity change of 0.1m/sec over the CPI, easily satisfying $\textbf{A3.}$. As for airborne targets, care must be taken to ensure compliance.

Based on these three assumptions, we can write the received signal as
\begin{equation}
\label{eq:x_received}
x(t) = \sum_{p=0}^{P-1} \sum_{\ell=0}^{L-1} \alpha_\ell h(t - \tau_\ell - p \tau) e^{-j \nu_\ell p \tau}.
\end{equation}
It will be convenient to express the signal as a sum of single frames
\begin{equation}
x(t) = \sum_{p=0}^{P-1} x_p(t)
\end{equation}
where
\begin{equation}
\label{eq:x_p}
x_p(t) = \sum_{\ell=0}^{L-1} \alpha_\ell h(t - \tau_\ell - p \tau) e^{-j \nu_\ell p \tau}.
\end{equation}
In reality $x(t)$ will be contaminated by AWGN. The main source of this noise is receiver thermal noise, but it also represents other wideband analog distortions and imperfections common to radio-frequency (RF) hardware. Another important source of distortion is clutter, originating from large objects in the targets' vicinity reflecting the transmitted signal. We take both into account in our analysis in Section~\ref{sec:DelayDopplerRecovery} and simulations in Section~\ref{sec:SimulationResults}.


Since the transmitted signal \eqref{eq:TxSignal} is a finite periodic pulse train, it is invariant to the transformation $\tau_\ell \rightarrow \tau_\ell + k_1 \tau, \,\, \nu_\ell \rightarrow \nu_\ell + 2\pi k_2 / \tau$ where $k_1,k_2 \in \mathbb{Z}$, except on its boundaries. Therefore the radar's unambiguous time-frequency region, where it can resolve targets with no ambiguity, is $[0,\tau]\times[-\pi/\tau, \pi/\tau]$ respectively. Targets outside this region will be measured in delay and Doppler frequency modulo $\tau$ and $2\pi/\tau$ accordingly, as with any fixed-PRI pulse radar. To resolve delay or Doppler ambiguity, information between several CPIs must be shared, for example by employing the Chinese remainder theorem \cite{ChineseRemainder}. We make the following further assumptions on targets' delay and Doppler:
\begin{enumerate}
\item[\textbf{A4.}] No time ambiguity: $\{\tau_\ell \in I \subset [0,\tau) \}_{\ell=0}^{L-1}$ where $I$ is a continuous time interval in $[0,\tau)$, so that $x_p(t)=0, \forall t \notin [p\tau,(p+1)\tau]$.
\item[\textbf{A5.}] No Doppler ambiguity: $\{\nu_\ell \in [-\pi/\tau,\pi/\tau) \}_{\ell=0}^{L-1}$.
\item[\textbf{A6.}] The pairs in the set $\{\tau_\ell, \nu_\ell\}_{\ell=0}^{L-1}$ are unique.
\end{enumerate}

Our goal in this work is to accurately detect the $L$ targets, \ie to estimate the $3L$ DOF $\{\alpha_\ell, \tau_\ell, \nu_\ell\}_{\ell=0}^{L-1}$ in \eqref{eq:x_received}, using the least possible number of digital samples.

\section{Classic Pulse-Doppler Processing}
\label{sec:ClassicProcessing}

Classic radar processing samples and processes the received signal at its Nyquist rate $B_h$ using a MF \cite{RadarHandbook}. In modern systems the MF operation is performed digitally, and therefore requires an ADC capable of sampling at $B_h$, which can be hundreds of MHz and even up to several GHz. In order to evaluate our sampling and reconstruction method, we compare it to classic radar processing, which in general consists of the following stages:
\begin{enumerate}
\item \textbf{ADC} -- sample each incoming frame $x_p(t)$ at its Nyquist rate $B_h$, equal to $h(t)$'s bandwidth, creating $x_p[n],0\leq n < N$, where $N=\tau B_h$.
\item \textbf{Matched filter} -- for each $x_p[n]$, create $y_p[n] = x_p[n] \ast h^\ast[-n]$, where $h[n]$ is a sampled version of the transmitted pulse $h(t)$. The time resolution attained in this step is $1/B_h$, corresponding to the width of the autocorrelation of the pulse $h(t)$.
\item \textbf{Doppler processing} -- for each discrete time index $n$, perform a $P$-point DFT along the pulse dimension: $z_n[k] = DFT_P\{y_p[n]\} = \sum_{p=0}^{P-1} y_p[n] e^{-j 2\pi p k/P}$ for $0\leq k < P$. The frequency resolution attained in this step is $2\pi/P\tau$, proportional to the inverse of the total coherent processing interval.
\item \textbf{Delay-Doppler map} -- stacking the vectors $\textbf{z}_n$, and taking absolute value, we obtain a delay-Doppler map $\textbf{Z} = \mbox{abs}[\textbf{z}_0 \; ... \; \textbf{z}_{N-1}] \in \mathbb{R}^{P \times N}$.
\item \textbf{Peak detection} -- a heuristic detection process, where knowledge of number of targets, target power, clutter location, etc. may help discover target positions. For example, if we know there are $L$ targets, then we can choose the $L$ strongest points in the map.
\end{enumerate}

The Doppler processing stage can be viewed as MF in the pulse dimension to a constant radial velocity target. As such, it increases the SNR by $P$ compared to the SNR of a single pulse \cite{Richards, NonCoherentGain}. Since a MF is the linear time-invariant (LTI) system which maximizes SNR, we know a factor $P$ increase is optimal for $P$ pulses. In Section~\ref{sec:DelayDopplerRecovery} we show that the SNR achieved with Doppler focusing also scales linearly with $P$, while other CS methods either achieve sub-linear SNR increase, or require prohibitive computational cost, as shown in Section~\ref{sec:PreviousApproaches}.

Classic processing requires sampling the received signal at its Nyquist rate $B_h$, which is inversely proportional to the system's time resolution. The required computational power is $P$ convolutions of a signal of length $N=\tau B_h$ and $N$ FFTs of length $P$ -- both also proportional to $B_h$. The growing demands for improved estimation accuracy and target separation dictate an ever growing increase in signal's bandwidth. The goal of this work is to \emph{present some concrete steps towards breaking the link between radar signal bandwidth and sampling rate}, and to allow low rate sampling and processing of radar signals, regardless of their bandwidth, retaining the same SNR scaling.

We achieve this goal by utilizing the combination of ideas of FRI and Xampling. Previous papers have already used these complementary concepts together. The work in \cite{multichannel, VetterliFRI} creates a mathematical framework for sub-Nyquist sampling of pulse streams and defines lower bounds on the sampling rate needed for perfect reconstruction. Practical sampling methods achieving these bounds are explained in \cite{EliGal, TurUltrasound, CompressedBeamformingUltrasound} in the context of ultrasound and radar, both without Doppler. Another work \cite{KfirWaheed} investigates the delay-Doppler estimation problem, but recovers the delays and Doppler frequencies in a two-stage process which achieves an SNR increase which is sub-linear in $P$. In Section~\ref{sec:DelayDopplerRecovery} we prove that in noisy scenarios, Doppler focusing recovery is superior to the two-stage method, as it achieves an SNR increase equal to a pulse dimension MF. We verify these results with simulations in Section~\ref{sec:SimulationResults}. Combining Xampling and Doppler focusing, our $3L$ DOF input signal \eqref{eq:x_received} will also enjoy the benefits of sub-Nyquist rate Xampling and accurate digital recovery of the target scene, effectively breaking the link between signal bandwidth and sampling rate.

\section{Doppler Focusing}
\label{sec:DopplerFocusing}

We now introduce and explain the main idea in this paper, called Doppler Focusing. This processing technique uses target echoes from different pulses to create a single superimposed pulse. This improves the SNR for robustness against noise and implicitly estimates targets' Doppler frequency in the process. We point out that stages 2) and 3) in classic processing can be viewed as delay and Doppler processing accordingly. Since they are both LTI, they can be interchanged, performing the DFT before MF. In classic processing, performing MF before DFT decreases computation latency, so most practical systems carry them out in the noted order. However, when using Doppler focusing, the Doppler focusing stage must come before delay estimation and the order can no longer be reversed.

When interchanging steps 2) and 3), the DFT is simply a discrete equivalent of the following time shift and modulation operation on the received signal:
\begin{align}
\label{eq:Phi}
\nonumber \Phi(t;\nu) &= \sum_{p=0}^{P-1} x_p(t+p\tau) e^{j \nu p\tau}\\
\nonumber &= \sum_{p=0}^{P-1} \sum_{\ell=0}^{L-1} \alpha_\ell h(t - \tau_\ell) e^{j (\nu - \nu_\ell) p\tau}\\
&= \sum_{\ell=0}^{L-1} \alpha_\ell h(t - \tau_\ell) \sum_{p=0}^{P-1} e^{j (\nu - \nu_\ell) p\tau}
\end{align}
where we used \eqref{eq:x_p}.

We now analyze the sum of exponents in \eqref{eq:Phi}. For any given $\nu$, targets with Doppler frequency $\nu_\ell$ in a band of width $2\pi/P\tau$ around $\nu$, \ie in $\Phi(t;\nu)'s$ ``focus zone", will achieve coherent integration and an SNR boost of approximately
\begin{equation}
\label{eq:g}
g(\nu | \nu_\ell) = \sum_{p=0}^{P-1} e^{j (\nu - \nu_\ell) p\tau}\overset{|\nu - \nu_\ell| < \pi/P\tau} \cong P
\end{equation}
compared with a single pulse. On the other hand, since the sum of $P$ equally spaced points covering the unit circle is generally close to zero, targets with $\nu_\ell$ not ``in focus" will approximately cancel out. Thus $g(\nu | \nu_\ell) \cong 0$ for $|\nu - \nu_\ell| > \pi/P\tau$, where using assumption \textbf{A5.} we assume $|\nu - \nu_\ell| < \pi / \tau$. See Fig.~\ref{fig:g_w} for an example of $g(\nu | \nu_\ell)$.
\begin{figure}[ht]
\centering
    \includegraphics[width=0.7 \columnwidth]{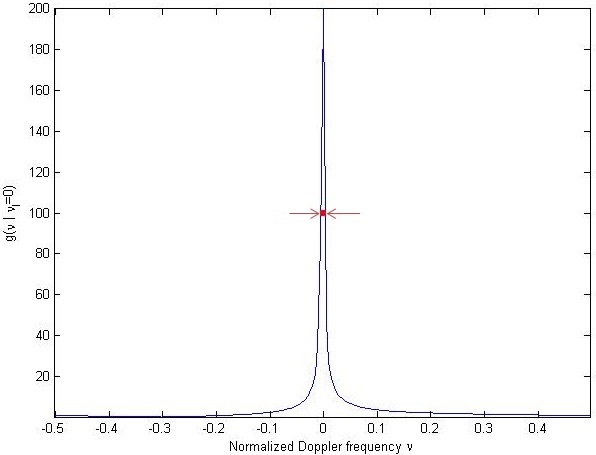}
\caption{Example of $g(\nu | \nu_\ell)$ for $P=200$ pulses and $\nu_\ell=0$. Arrows mark the ``focus zone", \ie $|\nu| < \pi/P\tau$. Frequencies outside focus zone are severely attenuated.}
\label{fig:g_w}
\end{figure}
Therefore we can approximate \eqref{eq:Phi} by
\begin{equation}
\label{eq:PhiApprox}
\Phi(t;\nu) \cong P \sum_{\ell \in \Lambda(\nu)} \alpha_\ell h(t - \tau_\ell)
\end{equation}
where $\Lambda(\nu) = \{ \ell : |\nu - \nu_\ell| < \pi/P\tau \}$.

Instead of trying to estimate delay and Doppler together, we have reduced our problem to delay only estimation for a small range of Doppler frequencies, with increased amplitude for improved performance against noise. To emphasize this, consider the case of trying to detect and estimate parameters for two targets with very closely spaced delays but with different Doppler frequencies (see for example the two helicopters in Fig.~\ref{fig:delay_Doppler_map}). Algorithms whose time resolution is coarser than the targets' delay separation are likely to encounter various problems recovering this target scene, the most likely of which is identification of a single target instead of two. With Doppler focusing we achieve an extra dimension of potential separation, regardless of the underlying delay recovery algorithm, enabling improved recovery performance. Fig.~\ref{fig:delay_Doppler_map} illustrates this concept by showing various targets spanning some delay-Doppler region. When focusing for some $\nu$ only targets in $\nu$'s focus zone (white region) come into view, while all other targets (red region) disappear. In Section~\ref{sec:SimulationResults} we demonstrate this point via simulation.
\begin{figure}[ht]
\centering
    \includegraphics[width=\columnwidth]{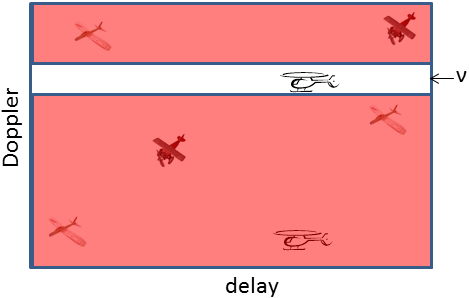}
\caption{Schematic delay-Doppler map. Red area indicates out-of-focus region. Only one target is in focus for current choice of $\nu$.}
\label{fig:delay_Doppler_map}
\end{figure}

To allow our sub-Nyquist recovery method, in the next sections we sample our signal in the time domain but extract frequency domain information. We now show how Doppler focusing can also be performed in the frequency domain, paving the way towards sub-Nyquist Doppler focusing.

Using \eqref{eq:x_p}, and denoting $X_p(\omega)$ as the CTFT of $x_p(t+p\tau)$,
\begin{equation}
X_p(\omega) = H(\omega) \sum_{\ell=0}^{L-1} \alpha_\ell e^{-j \omega \tau_\ell} e^{-j \nu_\ell p \tau}.
\end{equation}
Taking the CTFT of $\Phi(t;\nu)$ as a function of $t$ we obtain:
\begin{align}
\label{eq:FreqFocus}
\nonumber \Psi(\omega;\nu) &= \textup{CTFT}(\Phi(t;\nu)) = \sum_{p=0}^{P-1} X_p(\omega) e^{j \nu p\tau} \\
&= H(\omega) \sum_{\ell=0}^{L-1} \alpha_\ell e^{-j \omega \tau_\ell} \sum_{p=0}^{P-1} e^{j (\nu - \nu_\ell) p\tau}.
\end{align}
All $3L$ of the problem's parameters appear in \eqref{eq:FreqFocus}. Its structure is that of a delay estimation problem, as shown in Section~\ref{sec:DelayRecovery}, combined with the familiar sum of exponents term from \eqref{eq:Phi}.

We have seen that Doppler focusing reduces a delay-Doppler estimation problem to a delay-only estimation problem for a specific Doppler frequency. In the next section we describe delay recovery from sub-Nyquist sampling rates using Xampling. In Section~\ref{sec:DelayDopplerRecovery} we revisit the Doppler focusing concept, combining it with these low rate Xamples to address the delay-Doppler radar problem at sub-Nyquist rates.

\section{Sub-Nyquist Delay Recovery}
\label{sec:DelayRecovery}

The problem of recovering the $2L$ amplitudes and delays in
\begin{equation}
\label{eq:phi}
\phi(t) = \sum_{\ell=0}^{L-1} \alpha_\ell h(t-\tau_\ell), \quad 0 \leq t < \tau
\end{equation}
from sub-Nyquist samples has been previously studied in \cite{EliGal, multichannel, VetterliFRI, TurUltrasound}. Since Doppler focusing yields such a problem, we now review how Xampling can be used to solve \eqref{eq:phi} at a sub-Nyquist sampling rate.

\subsection{Xampling}

The concept of Xampling, introduced in \cite{Xampling, XamplingBook, XamplingSubspaces}, describes analog-to-digital conversion which acquires samples at sub-Nyquist rates while preserving the ability to perfectly reconstruct the signal. Xampling can be interpreted as ``compressed sampling", in the sense that we are performing data compression inherently in the sampling stage. To do this, we do not simply reduce sampling rate, since this is bound to cause loss of information. Instead, we perform an analog prefiltering operation on our signal and only then sample it, in order to extract the required information for recovery. We now show how the signal's Fourier series coefficients are related to the problem's unknown parameters \cite{multichannel, KfirWaheed, VetterliFRI, TurUltrasound}. We then describe how to acquire these coefficients via Xampling.

Since $\phi(t)$ is confined to the interval $t \in [0,\tau]$, it can be expressed by its Fourier series
\begin{equation}
\phi(t) = \sum_{k \in \mathbb{Z}} c[k] e^{j2\pi k t / \tau}, \quad t \in [0,\tau],
\end{equation}
where
\begin{align}
\label{eq:FourierCoeffs_2L}
\nonumber c[k] &= \frac{1}{\tau} \int_{0}^{\tau} \phi(t) e^{-j 2\pi k t / \tau} dt \\
\nonumber &= \frac{1}{\tau} \sum_{\ell=0}^{L-1} \alpha_\ell \int_{0}^{\tau} h(t-\tau_\ell) e^{-j 2\pi kt / \tau} dt \\
&= \frac{1}{\tau} H(2\pi k / \tau) \sum_{\ell=0}^{L-1} \alpha_\ell e^{-j 2\pi k \tau_\ell / \tau}.
\end{align}
From \eqref{eq:FourierCoeffs_2L} we see that the unknown parameters $\{\alpha_\ell, \tau_\ell\}_{\ell=0}^{L-1}$ are embodied in the Fourier coefficients $c[k]$ in the form of a complex sinusoid problem. For these problems, if there is no noise, $2L$ samples are enough to recover the unknown $\alpha$'s and $\tau$'s \cite{VetterliFRI}, \ie $|\kappa| \geq 2L$. We can solve this problem by using spectral analysis methods which require sampling a consecutive subset of coefficients, such as the annihilating filter \cite{annihilating}, matrix pencil \cite{pencil}, or ESPRIT \cite{ESPRIT}. An alternative approach is to use MUSIC \cite{MUSIC}, which does not require consecutive coefficients. The lower bound on $|\kappa|$ can be achieved only when the noise is negligible and computational complexity is not of concern. When there is substantial noise in the problem, having more than $2L$ coefficients will allow more robust recovery.

Our signals exist in the time domain, and therefore we do not have direct access to $c[k]$. We can use the Direct Multichannel Sampling scheme described in \cite{multichannel} in order to obtain any arbitrary set of Fourier series coefficients. Fig.~\ref{fig:multichannel} demonstrates a Xampling scheme used to directly extract the required Fourier coefficients from the signal. The analog input signal $x(t)$ is split into $|\kappa|$ channels, where in each channel $k$ it is mixed with the harmonic signal $e^{-j 2 \pi k t / \tau}$, integrated over the PRI duration, and then sampled.

In \cite{EliGal} an actual radar system was built using a similar yet more practical technique, where a set of mixers, band-pass filters and low rate ADCs sampled different spectral bands of the signal and the matching Fourier coefficients were created digitally. The radar model there included delay only without Doppler. An alternative Xampling method uses the Sum of Sincs filter as described in \cite{TurUltrasound}. All of these methods can be used to obtain arbitrary Fourier series coefficients. The number of Fourier coefficients extracted per pulse is a design parameter which controls the tradeoff between sampling rate and recovery performance. In our numerical experiments, presented in Section~\ref{sec:SimulationResults}, we demonstrate a Xampling rate one tenth of the Nyquist rate.
\begin{figure}[ht]
\centering
\includegraphics[width=\columnwidth]{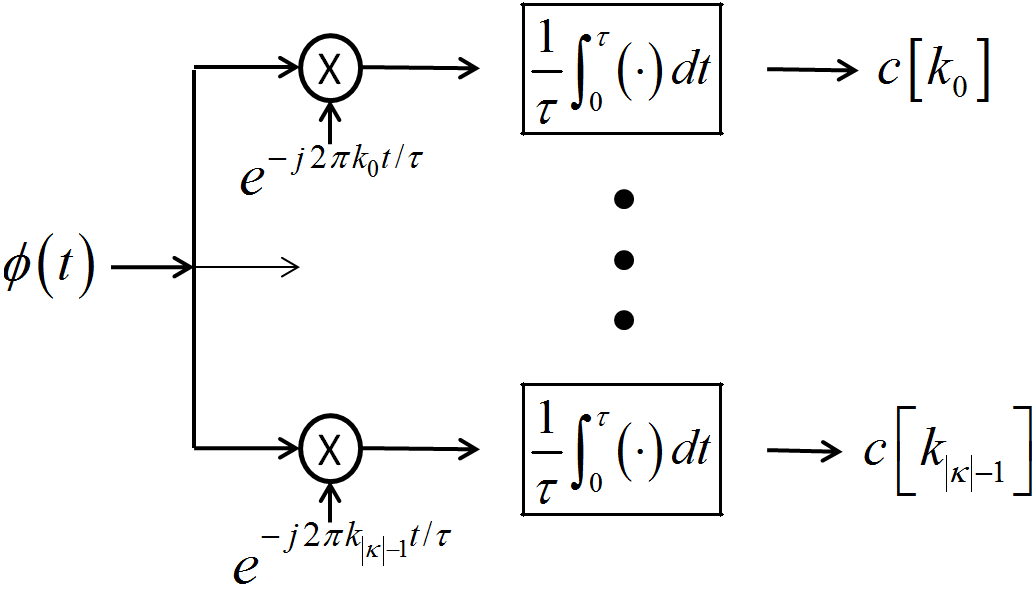}
\caption{Multichannel direct sampling of the Fourier series coefficients, from \cite{multichannel}.}
\label{fig:multichannel}
\end{figure}

\subsection{Compressed Sensing Recovery}

In the previous subsection we showed that $2L$ samples are enough to solve \eqref{eq:phi}, when there is no noise. We now describe a CS-based delay recovery method, operating on the Xamples $c[k]$, which is more robust to noise.

Assume the time delays are aligned to a grid
\begin{equation}
\tau_\ell=n_\ell\Delta_\tau, \quad 0 \leq n_\ell < N_\tau
\end{equation}
where we choose $\Delta_\tau$ so that $N_\tau = \tau/\Delta_\tau$ is an integer. Since in real scenarios target delays are not necessarily aligned to a grid, a local interpolation process around detected grid points can be used to reduce quantization errors. In our simulations in Section~\ref{sec:SimulationResults} we perform a parabolic fit around detected indices in order to improve delay estimate. The analysis in \cite{BasisMismatch} can be used to quantify these ``off-the-grid" errors. Choosing a set of indices $\kappa = \{ k_0, ..., k_{|\kappa|-1}\}$, we define the corresponding vector of Fourier coefficients
\begin{equation}
\label{eq:c_kappa}
\mathbf{c} = [c[k_0] \; ... \; c[k_{|\kappa|-1}]]^T \in \mathbb{C}^{|\kappa|}.
\end{equation}
We can then write \eqref{eq:FourierCoeffs_2L} in vector form as
\begin{equation}
\label{eq:CS_single}
\mathbf{c} = \frac{1}{\tau} \mathbf{H} \mathbf{V} \mathbf{x}
\end{equation}
where $\mathbf{H}$ is a $|\kappa| \times |\kappa|$ diagonal matrix with elements $H(2\pi k_i / \tau)$ and $\mathbf{V}$ is a $|\kappa| \times N_\tau$ Vandermonde matrix with $\mathbf{V}_{mq} = e^{-j 2\pi k_m n / N_\tau}$, \ie it is composed of $|\kappa|$ rows of the $N_\tau \times N_\tau$ DFT matrix. The target delay vector $\mathbf{x} \in \mathbb{C}^{N_\tau}$ is $L$-sparse, where each index $n$ contains the amplitude of a target with delay $n \Delta_\tau$ if it exists, or zero otherwise. Defining the CS dictionary $\mathbf{A} = \frac{1}{\tau}\mathbf{H} \mathbf{V} \in \mathbb{C}^{|\kappa| \times N_\tau}$ we obtain the CS equation
\begin{equation}
\label{eq:CS_single_nu_2L}
\mathbf{c} = \mathbf{A} \mathbf{x}.
\end{equation}
Estimating targets' delays can be carried out by solving \eqref{eq:CS_single_nu_2L} and finding $\mathbf{x}$'s support - any nonzero index $n$ denotes a target with delay $n \Delta_\tau$.

For any set of sampled Fourier coefficients, a variety of CS techniques can be employed for recovery \cite{CompressedBeamformingUltrasound}, for instance orthogonal matching pursuit (OMP) \cite{OMP}, iterative hard thresholding (IHT) \cite{IHT}, or L1 minimization (see \cite{CSTheoryApp} and references within). Choosing the coefficients at random produces favorable conditions for CS, aiding recovery in the presence of noise. When the indices in $\kappa$ are selected uniformly at random, it can be shown that if $|\kappa| \geq cL(\mbox{log}N_\tau)^4$, for some positive constant $c$, then $\mathbf{A}$ obeys the desired Restricted Isometry Property (RIP) with large probability \cite{rudelson2008sparse}. By satisfying the condition for RIP we are able to recover $\mathbf{x}$, using a CS recovery algorithm.

\section{Delay-Doppler Recovery}
\label{sec:DelayDopplerRecovery}

In Section~\ref{sec:DopplerFocusing} we introduced the concept of Doppler focusing, and in Section~\ref{sec:DelayRecovery} we reviewed how to Xample and recover the $2L$ unknowns of the delay estimation problem \eqref{eq:phi}. We now return to our original delay-Doppler problem \eqref{eq:x_received}.

We begin by describing how Xampling can be performed on the multi pulse signal \eqref{eq:x_received}. We then describe Doppler focusing, and analyze two aspects of the algorithm: the effect of multiple pulses on SNR when noise exists, and the minimal number of samples required for perfect recovery without noise. Finally we discuss some practical considerations.

\subsection{Xampling}

Similarly to the Xampling technique of Section~\ref{sec:DelayRecovery} which obtained $c[k]$, we can extend this technique to each of the pulses $x_p(t)$ of the multi-pulse signal \eqref{eq:x_received} to obtain $c_p[k]$. Since $x_p(t)$ is confined to the interval $t \in [p\tau,(p+1)\tau]$, we can replace $t \rightarrow t+p\tau$ and $\alpha_\ell \rightarrow \alpha_\ell e^{-j \nu_\ell p\tau}$ in \eqref{eq:FourierCoeffs_2L} to obtain
\begin{equation}
\label{eq:FourierCoeffs}
c_p[k] = \frac{1}{\tau} H(2\pi k / \tau) \sum_{\ell=0}^{L-1} \alpha_\ell e^{-j \nu_\ell p\tau} e^{-j 2\pi k \tau_\ell / \tau},
\end{equation}
where we used the fact that since both $k,p \in \mathbb{Z}$ we have $e^{-j 2\pi k p} \equiv 1$. Similarly to \eqref{eq:CS_single}, for each pulse $p$ we obtain
\begin{equation}
\label{eq:CS_p}
\mathbf{c}_p = \frac{1}{\tau} \mathbf{H} \mathbf{V} \mathbf{x}_p
\end{equation}

From \eqref{eq:FourierCoeffs} we see that all $3L$ unknown parameters $\{\alpha_\ell, \tau_\ell, \nu_\ell\}_{\ell=0}^{L-1}$ are embodied in the Fourier coefficients $c_p[k]$ in the form of a complex sinusoid problem. The number of Fourier coefficients sampled in each pulse, $|\kappa|$, controls the trade-off between sample rate and robustness to noise.

\subsection{Applying Doppler Focusing and CS Recovery}

Having acquired $c_p[k]$ using Xampling, we now perform the Doppler focusing operation for a specific frequency $\nu$
\begin{align}
\label{eq:PsiDef}
\nonumber & \Psi_\nu[k] = \sum_{p=0}^{P-1} c_p[k] e^{j \nu p \tau} \\
& = \frac{1}{\tau} H(2\pi k / \tau) \sum_{\ell=0}^{L-1} \alpha_\ell e^{-j 2\pi k \tau_\ell / \tau} \sum_{p=0}^{P-1} e^{j (\nu-\nu_\ell) p \tau}.
\end{align}
From \eqref{eq:FreqFocus} we see that $\Psi_\nu[k] = \tau \Psi(\omega;\nu)|_{\omega = 2\pi k / \tau}$.

Following the same arguments as in \eqref{eq:g}, for any target $\ell$ satisfying $|\nu-\nu_\ell| < \pi/P\tau$ we have
\begin{align}
\sum_{p=0}^{P-1} e^{j (\nu-\nu_\ell) p \tau} \cong P.
\end{align}
Therefore, Doppler focusing can be performed on the low rate sub-Nyquist samples:
\begin{equation}
\label{eq:coeff_single}
\Psi_\nu[k] \cong \frac{P}{\tau} H(2\pi k / \tau) \sum_{\ell \in \Lambda(\nu)} \alpha_\ell e^{-j 2\pi k \tau_\ell / \tau}.
\end{equation}
Equation \eqref{eq:coeff_single} is identical in form to \eqref{eq:FourierCoeffs_2L} except it is scaled by $P$, increasing SNR for improved performance with noise. Furthermore, we reduced the number of active delays.

For each $\nu$ we now have a delay estimation problem, which can be written in vector form using the same notations of Section~\ref{sec:DelayRecovery} as
\begin{equation}
\label{eq:CS_single_nu}
\mathbf{\Psi}_\nu = \frac{P}{\tau} \mathbf{H} \mathbf{V} \mathbf{x}_\nu
\end{equation}
where $\mathbf{x}_\nu$ is $L$-sparse and
\begin{equation}
\label{eq:PsiVec}
\mathbf{\Psi}_\nu = [\Psi_\nu[k_0] \; ... \; \Psi_\nu[k_{|\kappa|-1}]]^T \in \mathbb{C}^{|\kappa|}.
\end{equation}
This is exactly the CS problem we have already shown how to solve in Section~\ref{sec:DelayRecovery}. Here it is important to note that our dictionary is indifferent to the Doppler estimation. CS methods which estimate delay and Doppler simultaneously \cite{HighResRadar, AdaptiveCSRadar}, require a dictionary which grows with the number of pulses. Here by separating delay and Doppler estimation, the CS dictionary is not a function of $P$.

The Doppler focusing operation \eqref{eq:PsiDef} is a continuous operation on the variable $\nu$, and can be performed for any Doppler frequency up to the PRF. With Doppler focusing there are no inherent ``blind speeds", \ie target velocities which are undetectable, as occurs with classic Moving Target Indication (MTI) \cite{Richards}. Define the set of Fourier coefficients $C = \{c_p[k]\}_{0 \leq p < P}^{k \in \kappa}$, and $\mathbf{\Psi}_\nu(C)$ as the vector of focused coefficients \eqref{eq:PsiVec} obtained from $C$ using \eqref{eq:PsiDef}. Therefore $\mathbf{x}_\nu(C)$, where again we explicitly emphasize the dependence on the set $C$, can be recovered from $\mathbf{\Psi}_\nu(C)$ for any $\nu$. Since strong amplitudes are indicative of true target existence as opposed to noise, Doppler focusing recovery searches for large values of $|\mathbf{x}_\nu(C)[n]|$ and estimates target delays and Doppler frequencies as $n \Delta_\tau$ and $\nu$ accordingly. After detecting each target, its influence is removed from the set of Fourier coefficients in order to reduce masking of weaker targets and to remove spurious targets created by processing sidelobes. A similar subtraction is performed in many iterative algorithms such as OMP or the Clean Process of \cite{CLEAN}. Detection is performed iteratively until all targets have been detected, if $L$ is known, or until an amplitude threshold is met, if the model order is unknown. Algorithm~\ref{alg:DopplerFocusing} summarizes the Doppler focusing algorithm for the case of known $L$.

\begin{algorithm}
\caption{Doppler Focusing}
\label{alg:DopplerFocusing}
\textbf{Input:} Xamples $C = \{c_p[k]\}_{0 \leq p < P}^{k \in \kappa}$, number of targets $L$ \\
\textbf{Output:} Estimated target parameters $\{\hat{\alpha}_\ell, \hat{\tau}_\ell, \hat{\nu}_\ell\}_{\ell=0}^{L-1}$ \\
\begin{algorithmic}
\STATE \textbf{Initialization:} $R = \{r_p[k]\}_{0 \leq p < P}^{k \in \kappa} \leftarrow C$ \\
\FOR{$\ell = 0$ to $L-1$}
\STATE $(\hat{n}_\ell, \hat{\nu}_\ell) \leftarrow \underset{\substack{0 \leq \hat{n} < N_\tau \\
-\pi/\tau \leq \hat{\nu} < \pi/\tau}}
{\arg\max} |\mathbf{x}_{\hat{\nu}}(R)[\hat{n}]|$ using \eqref{eq:CS_single_nu}
\STATE $\hat{\tau}_\ell \leftarrow \hat{n}_\ell \Delta_\tau$
\STATE $\hat{\alpha}_\ell \leftarrow \mathbf{x}_{\hat{\nu}_\ell}(R)[\hat{n}_\ell]$
\FOR{$k \in \kappa$ and $0 \leq p < P$}
\STATE $r_p[k] \leftarrow r_p[k] - \frac{1}{\tau} H(2\pi k / \tau) \hat{\alpha}_\ell e^{-j \hat{\nu}_\ell p\tau} e^{-j 2\pi k \hat{\tau}_\ell / \tau}$
\ENDFOR
\ENDFOR
\end{algorithmic}
\end{algorithm}

In the delay-only problem of Section~\ref{sec:DelayRecovery}, the model order $L$ was known. Here, since there are $L$ targets but we have no prior knowledge of their distribution in the delay-Doppler plane, each time we solve \eqref{eq:CS_single_nu} we must either estimate the model order $0 \leq L_\nu \leq L$, or take a worst case approach and assume $L_\nu = L$. The problem of estimating the number of sinusoids in a noisy sequence has been studied extensively \cite{GSchwartz, WaxKailath, JJFuchs}. Solving \eqref{eq:CS_single_nu} with an accurate model order can decrease computation time (although estimating model order is also time consuming) and possibly reduce detection of spurious targets. The former option is preferable in higher SNR scenarios when the $L_\nu$ estimates are good, while the latter can be a fall back approach in noisy scenarios when the $L_\nu$ estimates contain significant error. In our simulations, since we wish to eliminate model order errors which influence recovery performance, we employ the worst case approach.

Since $\mathbf{V}$ is a partial Fourier matrix, the problem defined in \eqref{eq:CS_single_nu}, after normalizing by $\mathbf{H}^{-1}$, becomes a problem of recovering frequencies from a sum of complex exponentials. Many methods exist for solving such a spectral analysis problem (see \cite{annihilating} for a review), and they can be used instead of CS. These methods offer different combinations of robustness to noise, minimal sample rate, and sensitivity to grid errors. CS is our method of choice when we are interested in low-SNR scenarios. However, since Doppler focusing is independent of the underlying delay estimation, in different scenarios CS can be exchanged for alternative delay recovery methods. For example, in the upcoming noiseless recovery subsection, when noise is not a concern, we use the annihilating filter approach instead of CS.

\subsection{SNR Analysis}

To analyze the effect of Doppler focusing on SNR, we add noise to \eqref{eq:x_received}:
\begin{equation}
\tilde{x}(t) = x(t) + w(t),
\end{equation}
where $w(t)$ is a zero mean wide-sense stationary random signal with autocorrelation $r_w(s) = \sigma^2 \delta(s)$. The Fourier coefficients in \eqref{eq:FourierCoeffs} then become
\begin{equation}
\tilde{c}_p[k] = c_p[k] + w_p[k],
\end{equation}
where \begin{equation}
w_p[k] = \frac{1}{\tau} \int_{p \tau}^{(p+1) \tau} w(t) e^{-j 2\pi k t / \tau} dt
\end{equation}
is a zero mean complex random variable with variance
\begin{align}
\label{eq:variance}
\nonumber &E \left[ \left| w_p[k] \right|^2 \right] \\
\nonumber &= \frac{1}{\tau^2} \int_{p \tau}^{(p+1) \tau} dt \int_{p \tau}^{(p+1) \tau} E[w(t)w^*(t')] e^{-j 2\pi k (t-t') / \tau} dt' \\
\nonumber &= \frac{1}{\tau^2} \int_{p \tau}^{(p+1) \tau} dt \int_{p \tau}^{(p+1) \tau} \sigma^2 \delta(t-t') e^{-j 2\pi k (t-t') / \tau} dt' \\
&= \frac{1}{\tau^2} \sigma^2 \int_{p \tau}^{(p+1) \tau} dt = \sigma^2 / \tau.
\end{align}
We can write $\tilde{c}_p[k]$ as the disjoint effect of $L$ targets and noise
\begin{equation}
\label{eq:noisy_coeff}
\tilde{c}_p[k] = \sum_{\ell=0}^{L-1} c_p^\ell[k] + w_p[k]
\end{equation}
where $c_p^\ell[k] = \frac{1}{\tau} H(2\pi k / \tau) \alpha_\ell e^{-j \nu_\ell p\tau} e^{-j 2\pi k \tau_\ell / \tau}$. Therefore we can define the SNR of target $\ell$ in $\tilde{c}_p[k]$ as the power ratio
\begin{align}
\label{eq:SNR1}
\nonumber \Gamma_p^\ell[k] &= \frac{\left| c_p^\ell[k] \right|^2}{E \left[ \left| w_p[k] \right|^2 \right]} \\
\nonumber &= \frac{|\frac{1}{\tau} H(2\pi k / \tau) \alpha_\ell e^{-j \nu_\ell p\tau} e^{-j 2\pi k \tau_\ell / \tau}|^2}{\sigma^2 / \tau} \\
 &= |H(2\pi k / \tau)|^2 \frac{|\alpha_\ell|^2}{\tau \sigma^2}.
\end{align}

Analyzing the SNR in the focused Fourier coefficients using \eqref{eq:PsiDef} we obtain
\begin{align}
\label{eq:SNR2}
\nonumber \tilde{\Psi}_\nu[k] &= \sum_{p=0}^{P-1} \tilde{c}_p[k] e^{j \nu p \tau} \\
\nonumber &= \sum_{p=0}^{P-1} \left( c_p[k] + w_p[k] \right) e^{j \nu p \tau} \\
&= \Psi_\nu[k] + w_\nu[k]
\end{align}
where $w_\nu[k] = \sum_{p=0}^{P-1} w_p[k] e^{j \nu p \tau}$. Using \eqref{eq:coeff_single} we have
\begin{align}
\label{eq:SNR3}
\nonumber & \Psi_\nu[k] + w_\nu[k] \\
\nonumber &\cong \sum_{\ell \in \Lambda(\nu)} \frac{P}{\tau} H(2\pi k / \tau) \alpha_\ell e^{-j 2\pi k \tau_\ell / \tau} + w_\nu[k] \\
&= \sum_{\ell \in \Lambda(\nu)} \Psi_\nu^\ell[k] + w_\nu[k],
\end{align}
where $\Psi_\nu^\ell[k] = \frac{P}{\tau} H(2\pi k / \tau) \alpha_\ell e^{-j 2\pi k \tau_\ell / \tau}$. We then define the focused SNR for target $\ell$ as
\begin{equation}
\label{eq:SNR4}
\Xi_\nu^\ell[k] = \frac{\left| \Psi_\nu^\ell[k] \right|^2}{E \left[ \left| w_\nu[k] \right|^2 \right]}.
\end{equation}

Due to the independence of $w(t)$ in different time intervals, $w_p[k]$ is a discrete time white random sequence with respect to $p$, so $w_\nu[k]$ is a sum of $P$ independent random variables. Therefore, using \eqref{eq:variance}, we have
\begin{equation}
\label{eq:FocusedVariance}
E \left[ \left| w_\nu[k] \right|^2 \right] = \sum_{p=0}^{P-1} |e^{j \nu p \tau}|^2 E \left[ \left| w_p[k] \right|^2 \right] = P \sigma^2 / \tau.
\end{equation}
Substituting \eqref{eq:FocusedVariance} into \eqref{eq:SNR4} we obtain
\begin{align}
\label{eq:SNR5}
\nonumber \Xi_\nu^\ell[k] &= \frac{\left| \frac{P}{\tau} H(2\pi k / \tau) \alpha_\ell e^{-j 2\pi k \tau_\ell / \tau} \right|^2}{P \sigma^2 / \tau} \\
&= P |H(2\pi k / \tau)|^2 \frac{|\alpha_\ell|^2}{\tau \sigma^2} = P \Gamma_p^\ell[k].
\end{align}

It is evident that the focused SNR \eqref{eq:SNR5} is $P$ times greater than before Doppler focusing \eqref{eq:SNR1}. We have obtained a linear SNR improvement with increasing number of pulses, as does an optimal MF.

\subsection{Noiseless Recovery}

The following theorems analyzes the minimal number of samples required for perfect recovery when there is no noise.

\begin{theorem}
\label{thm:noiselessrecovery_continuous}
The minimal number of samples required for perfect recovery of $L$ targets when there is no noise, is at least $4L^2$, with $|\kappa|$ and $P$ at least $2L$ each.
\end{theorem}

\begin{proof}
Recall \eqref{eq:FourierCoeffs} and denote normalized delay $q_\ell = 2 \pi \tau_\ell / \tau$ and normalized Doppler frequency $g_\ell = \nu_\ell \tau$, to obtain a more symmetric form
\begin{equation}
\label{eq:FourierCoeffsSymmetric}
c_p[k] = \frac{1}{\tau} H(2\pi k / \tau) \sum_{\ell=0}^{L-1} \alpha_\ell e^{-j (g_\ell p + q_\ell k)}.
\end{equation}
Since there are no constraints on either delay or Doppler frequency, let us examine the case where all targets have identical Doppler $g_\ell = g$:
\begin{equation}
\label{eq:FourierCoeffsSymmetricReduced1}
c_p[k] = \frac{1}{\tau} H(2\pi k / \tau) e^{-j g p} \sum_{\ell=0}^{L-1} \alpha_\ell e^{-j q_\ell k}.
\end{equation}
If it were possible to solve \eqref{eq:FourierCoeffsSymmetricReduced1} with less than $|\kappa| = 2L$ samples, for example by utilizing information from different values of $p$, then we could use this to bring \eqref{eq:FourierCoeffs_2L} to the form of \eqref{eq:FourierCoeffsSymmetricReduced1} by multiplying $c[k]$ by arbitrary values of $e^{-j g p}$. Thus we could solve \eqref{eq:phi} with less than $2L$ samples, in contradiction with \cite{VetterliFRI}. Therefore $|\kappa| \geq 2L$.

Inspecting the case where all targets have the same delay $q_\ell = q$, we obtain:
\begin{equation}
\label{eq:FourierCoeffsSymmetricReduced2}
c_p[k] = \frac{1}{\tau} H(2\pi k / \tau) e^{-j q k} \sum_{\ell=0}^{L-1} \alpha_\ell e^{-j g_\ell p}.
\end{equation}
Applying the same logic which deduced $|\kappa| \geq 2L$ from \eqref{eq:FourierCoeffsSymmetricReduced1}, we infer $P \geq 2L$ from \eqref{eq:FourierCoeffsSymmetricReduced2}.
\end{proof}

\begin{theorem}
\label{thm:noiselessrecovery_grid}
Suppose target Doppler frequencies are aligned to a grid $\{\tilde{\nu}_m = 2 \pi m / \tau M\}_{m=-M/2}^{M/2-1}$, with no restriction on target delays. Then the minimal number of samples required for perfect recovery of $L$ targets when there is no noise, is $2L \min(M,2L)$.
\end{theorem}

\begin{proof}
For any $k \in \kappa$ we use \eqref{eq:FourierCoeffs} to write $c_p[k] = \frac{1}{\tau} H(2\pi k / \tau) \sum_{\ell=0}^{L-1} \beta_\ell^k e^{-j \nu_\ell p\tau}$ where $\beta_\ell^k = \alpha_\ell e^{-j 2\pi k \tau_\ell / \tau}$. We obtain a standard CS problem by writing $\mathbf{c}[k] = \frac{1}{\tau} \mathbf{H} \mathbf{F} \mathbf{\beta}[k]$, where $\mathbf{H}$ is as defined in \eqref{eq:CS_single}, $\mathbf{F} \in \mathbb{C}^{P \times M}$ has elements $\mathbf{F}_{pm} = e^{-j 2 \pi m p / M}$, and $\mathbf{\beta}[k] \in \mathbb{C}^M$ is an $L$-sparse vector.

If $M \geq 2L$ then there is no gain compared with the continuous setting of Theorem~\ref{thm:noiselessrecovery_continuous}, and the minimal number of samples remains $4L^2$.

On the other hand, if $M < 2L$, then for $P \geq M$ the CS system is overdetermined and can be solved with the pseudoinverse $\mathbf{\beta[k]} = (\mathbf{A}^\textit{H} \mathbf{A})^{-1} \mathbf{A}^\textit{H} \mathbf{c[k]}$, where $\mathbf{A} = \frac{1}{\tau} \mathbf{H} \mathbf{F}$. For different values of $k$, each element $m$ of $\mathbf{\beta[k]}$ describes a spectral analysis problem for Doppler frequency $2 \pi m / \tau M$, containing no more than $L$ harmonies. Thus, we require $|\kappa| \geq 2L$ to complete the recovery.
\end{proof}

\begin{theorem}
\label{thm:noiselessrecovery_DopplerFocusing}
Under the conditions of Theorem~\ref{thm:noiselessrecovery_grid}, the minimal number of samples required for perfect recovery of $L$ targets using Doppler focusing is $2LM$, with $|\kappa| \geq 2L$ and $P \geq M$.
\end{theorem}

\begin{proof}
Sample $|\kappa|=2L$ Fourier coefficients for $P=M$ pulses, and choose a continuous set $\kappa$ so that $H(2 \pi k / \tau) \neq 0$ for all $k \in \kappa$. Next perform Doppler focusing for each grid point $\tilde{\nu}_m$ using \eqref{eq:PsiDef}, and solve the resulting delay estimation problem each time using an annihilating filter. In this setting, the analysis of Section~\ref{sec:DopplerFocusing}, and specifically \eqref{eq:g}, \eqref{eq:PhiApprox} are exact:
\begin{align}
\nonumber g(\tilde{\nu}_m | \nu_\ell) &= \sum_{p=0}^{M-1} e^{j (\tilde{\nu}_m - \nu_\ell) p \tau}  \\
\nonumber &= \sum_{p=0}^{M-1} e^{j (2 \pi m / \tau M - 2 \pi m_\ell / \tau M) p \tau} \\
&= \sum_{p=0}^{M-1} e^{j 2 \pi (m - m_\ell) p M} =
\begin{cases}
    M,& \text{if } m = m_\ell \\
    0,& \text{otherwise}
\end{cases}
\end{align}
where $m_\ell \in [-M/2, \ldots, M/2-1]$ satisfies $\nu_\ell = \tilde{\nu}_{m_\ell}$.

Each target $\ell$ appears in exactly one delay estimation problem, for $m_\ell$ satisfying $\nu_\ell = \tilde{\nu}_{m_\ell}$, and completely cancels out for all other $M-1$ problems. Every delay estimation problem \eqref{eq:CS_single_nu}, after normalizing by $\mathbf{H}^{-1}$, becomes a spectral analysis problem with $|\kappa|$ samples. With no prior information regarding target distribution in Doppler, we must assume the worst case of all $L$ targets occupying a single Doppler frequency, so that any spectral analysis problem has no more than $L$ frequencies. When there is no noise, the annihilating filter requires $|\kappa| \geq 2L$ samples for perfect recovery \cite{VetterliFRI}, and the total number of samples comes to $|\kappa|P = 2LM$.
\end{proof}


The minimal rate requirement exists separately on the number of sampled Fourier coefficients $|\kappa|$ and the number of sampled pulses $P$, and not for their product. This shows that in terms of minimal sampling rate, samples in the coefficient dimension $k$ cannot be replaced by samples in the pulse dimension $p$, and vice versa.

These theorems show that the requirement of Doppler focusing for $|\kappa| \geq 2L$ matches the general lower bound on the number of samples required in each pulse. Furthermore, when $M = O(L)$, the number of pulses required for Doppler focusing is within order of magnitude of the lower bound. And finally, the result in Theorem~\ref{thm:noiselessrecovery_continuous} coincides with the minimal sampling rate for two dimensional spectral analysis \cite{TwoDimMatrixPencil}.

\subsection{Practical Considerations}

We now describe a few practical issues, starting with computational efficiency.

If one wishes to perform the $\arg\max$ step of Algorithm~\ref{alg:DopplerFocusing} by probing a uniform grid of $M$ Doppler frequencies, \ie $\{\tilde{\nu}_m = 2 \pi m / \tau M\}_{m=-M/2}^{M/2-1}$, then $\Psi_\nu[k]$ can be created efficiently using a length $M$ DFT or FFT of a length $P$ series:
\begin{align}
\label{eq:FFT_focusing}
\nonumber \Psi_m[k] & \triangleq \Psi_{\tilde{\nu}_m}[k] = \sum_{p=0}^{P-1} c_p[k] e^{j \tilde{\nu}_m p \tau} \\
&= \sum_{p=0}^{P-1} c_p[k] e^{j 2\pi m p / M} = DFT_{M}\{c_p[k]\}.
\end{align}
Algorithm~\ref{alg:DopplerFocusingWithGrid} describes Doppler focusing performed with Doppler frequencies lying on a uniform grid.

\begin{algorithm}
\caption{Doppler Focusing with Grid}
\label{alg:DopplerFocusingWithGrid}
\textbf{Input:} Xamples $C = \{c_p[k]\}_{0 \leq p < P}^{k \in \kappa}$, number of targets $L$, Doppler grid size $M$ \\
\textbf{Output:} Estimated target parameters $\{\hat{\alpha}_\ell, \hat{\tau}_\ell, \hat{\nu}_\ell\}_{\ell=0}^{L-1}$ \\
\begin{algorithmic}
\STATE \textbf{Initialization:} $R = \{r_p[k]\}_{0 \leq p < P}^{k \in \kappa} \leftarrow C$ \\
\FOR{$\ell = 0$ to $L-1$}
\STATE Create $\{ \mathbf{\Psi}_m \}_{m=-M/2}^{M/2-1}$ from $R$ using FFT \eqref{eq:FFT_focusing}.
\STATE $(\hat{n}_\ell, \hat{m}_\ell) \leftarrow \underset{\substack{0 \leq \hat{n} < N_\tau \\
0 \leq \hat{m} < M }}{\arg\max} |\mathbf{x}_{\hat{m}}(R)[\hat{n}]|$ using \eqref{eq:CS_single_nu}
\STATE $\hat{\tau}_\ell \leftarrow \hat{n}_\ell \Delta_\tau$
\STATE $\hat{\nu}_\ell \leftarrow 2 \pi \hat{m}_\ell / \tau M $
\STATE $\hat{\alpha}_\ell \leftarrow \mathbf{x}_{\hat{m}_\ell}(R)[\hat{n}_\ell]$
\FOR{$k \in \kappa$ and $0 \leq p < P$}
\STATE $r_p[k] \leftarrow r_p[k] - \frac{1}{\tau} H(2\pi k / \tau) \hat{\alpha}_\ell e^{-j \hat{\nu}_\ell p\tau} e^{-j 2\pi k \hat{\tau}_\ell / \tau}$
\ENDFOR
\ENDFOR
\end{algorithmic}
\end{algorithm}

Another practical concern is target dynamic range. Since target amplitudes can differ by several orders of magnitude, care must be taken so strong targets do not mask weaker ones. When focusing on some Doppler frequency $\nu$, targets with Doppler frequencies $\nu_\ell$ satisfying $|\nu_\ell - \nu| > \pi/P\tau$ are considered undesirable and we seek to minimize their effect. These targets can be viewed as ``out-of-focus", since they are not matched to $\nu$ and their responses from different pulses do not combine coherently; they will combine in phase for different $\nu$'s satisfying $|\nu_\ell - \nu| < \pi/P\tau$. We can add to \eqref{eq:PsiDef} a user defined window function $w[p], p=0,1,...,P-1$ (\eg Hann, Blackman, etc.) which is designed to mitigate the impact of these out-of-focus targets:
\begin{align}
\label{eq:Windowing}
\nonumber & \Psi_\nu[k] = \sum_{p=0}^{P-1} c_p[k] e^{j \nu p \tau} w[p] \\
& = \frac{1}{\tau} H(2\pi k / \tau) \sum_{\ell=0}^{L-1} \alpha_\ell e^{-j 2\pi k \tau_\ell / \tau} \sum_{p=0}^{P-1} e^{j (\nu-\nu_\ell) p \tau} w[p].
\end{align}

The drawback of windowing is that it increases the frequency's focus zone, potentially including more targets in each delay estimation problem. In Fig.~\ref{fig:windows} we see an example of how windowing can reduce the effect of out-of-focus targets compared with no windowing (constant $w[p]$). For a comprehensive review of windowing function design considerations see \cite{Windows}. When attempting to support a set of targets amplitudes comprising a large dynamic range, a situation fairly common in real scenarios, the focusing operation must be performed with aggressive windowing in order for the strong targets to be sufficiently attenuated. In Section~\ref{sec:SimulationResults} we show a simulation demonstrating the improved dynamic range attained using Doppler focusing.

\begin{figure}[ht]
\centering
    \includegraphics[width=0.8 \columnwidth]{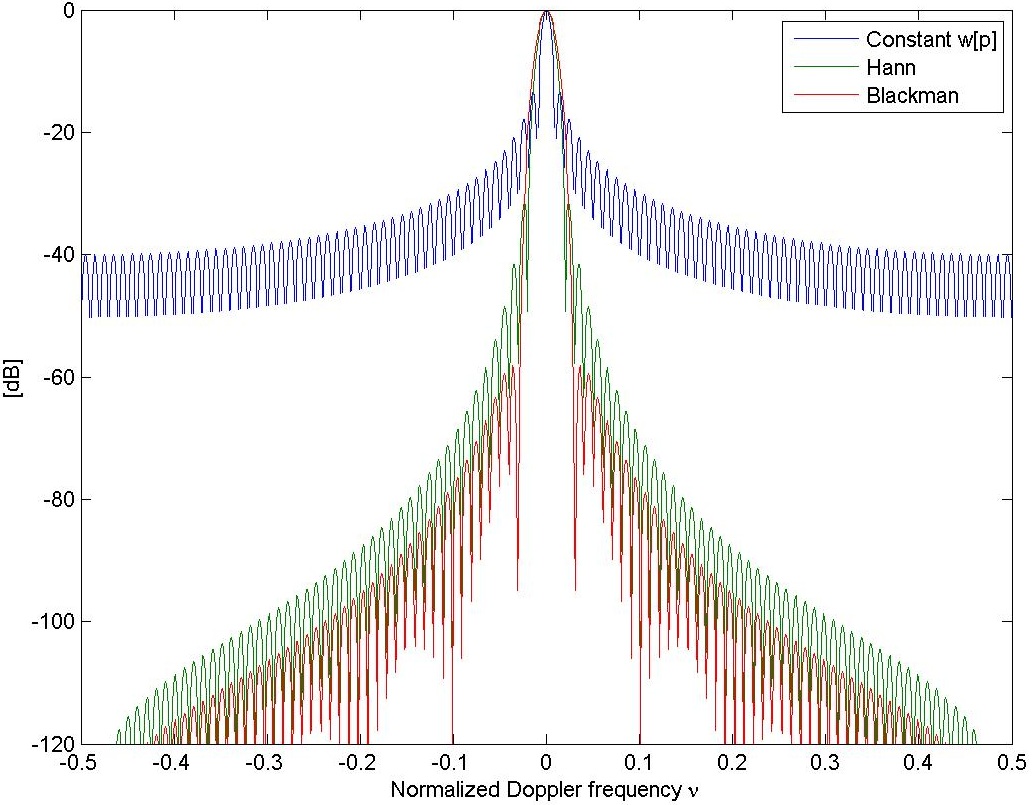}
\caption{DFT of windowing functions $w[p]$ compared with no windowing (constant $w[p]$) for $P=100$ pulses. Attenuation of targets with Doppler frequencies far enough from nominal frequency ($\nu=0$ here) increases significantly with proper windowing. Also, focus zone changes for different $w[p]$'s.}
\label{fig:windows}
\end{figure}

So far we assumed only Swerling-0 point targets. We briefly comment on targets having nonzero Doppler spread - \ie the micro-Doppler phenomenon \cite{Micro-Doppler}. This effect is the modulation of the target's main Doppler frequency caused by motion of the structure of the target around its main trajectory. For example, missile wingtips will exhibit different Doppler frequencies compared with the missile body when the missile is rotating around it central axis. Also, the missile body itself, when undergoing strong vibrations, can generate a spectrum of Doppler frequencies around its central velocity.

Our concern when considering this effect is whether the sparsity assumption remains valid. At the heart of our model is the assumption that the target scene is composed of a small number of targets with discrete Doppler frequencies. Recent works \cite{Micro-Doppler, Micro-DopplerSignatures} show that for cases of micro-Doppler, there are a small number of dominant frequencies in the continuous Doppler spectrum, caused by distinct vibration modes, rotation rates or resonant frequencies. Therefore, in most cases, targets exhibiting micro-Doppler can be treated as a superposition of several closely spaced targets, and the sparse target scene assumption remains valid.


\subsection{Clutter}

To complete this section, we add an analysis of the effects of clutter, considered the next major source of interference after thermal noise. Doppler focusing appears to have inherent clutter rejection capabilities, suggesting that special prefiltering operations such as MTI may not be required.

Clutter refers to unwanted echoes reaching the radar receiver from objects which are not the sought after targets: land, sea, buildings, etc. Due to the relative sizes of the objects, clutter echoes are usually several orders of magnitude stronger than target echoes, so if not treated properly, they can mask the target signals and prevent detection. Furthermore, as opposed to random noise which can be mitigated using coherent integration as in \eqref{eq:SNR5}, clutter echoes are (deterministic) scaled, shifted and modulated replicas of the transmitted signal. They ``enjoy" the benefits of coherent integration in exactly the same way as the target signal does. Therefore, the signal-to-clutter ratio (SCR) cannot be increased by increasing the CPI.

The most common method to allow detection in clutter ridden scenarios, is to utilize the fact that clutter, as opposed to most targets, is mostly static. If we assume the radar transceiver itself is also stationary, then clutter echoes will be received with zero Doppler frequency. This is the reason that classic anti-clutter methods (\eg MTI) are basically a notch filter blocking the Doppler frequency generated by the radar's own motion. We now show that Doppler focusing includes inherent target-clutter Doppler separation, so it does not require any prefiltering or modifications in order to allow target detection when facing clutter.

The Doppler focusing operation \eqref{eq:PsiDef} can be viewed as passing the Xamples $c_p[k]$ through a bandpass filter bank, where each filter has a pass-band of width $2 \pi / P \tau$. The filters' attenuation can be controlled using windowing \eqref{eq:Windowing}, at the cost of increasing the pass-band width. This creates adjustable isolation between delay estimation problems \eqref{eq:CS_single_nu} for targets with Doppler frequencies separated by more than the pass-band width. Therefore, if clutter were to be primarily concentrated around some specific frequency, targets with Doppler frequencies shifted away by more than approximately $2 \pi / P \tau$ could be detected without interference. This emphasizes that increasing the number of pulses $P$ can be used to improve isolation between targets and clutter. In Section~\ref{sec:SimulationResults} we demonstrate via simulation successful target detection using Doppler focusing in a clutter ridden scenario.

\section{Comparison to Previous Approaches}
\label{sec:PreviousApproaches}

In the previous section we described Doppler focusing and analyzed various aspects of its performance. We now compare this performance to two other methods for delay-Doppler estimation.

\subsection{Simultaneous Delay-Doppler Recovery}

Various methods exist which perform simultaneous delay-Doppler estimation \cite{HighResRadar, AdaptiveCSRadar}. These methods discretize the delay-Doppler plane, and construct a CS dictionary with a column for each two dimensional grid point. This approach can be seen as a discrete MF, where as opposed to classic processing, matching is performed to the entire pulse train rather than to a single pulse. Like a MF, these methods perform well in noisy conditions, since they achieve an SNR which scales linearly with $P$.

The critical drawback of this approach is that for any realistic problem size, the dictionary size grows rapidly and becomes too large to store or process. For even moderate size problems, the number of delay or Doppler grid points can easily be on the order of $10^3$, with a number of measurements of similar order. This requires the dictionary to posses $10^9$ elements, and occupy many gigabytes of memory, which requires high-end processing capabilities. Doppler focusing enjoys the same SNR improvement as simultaneous delay-Doppler recovery, with a dictionary whose size depends on delay estimation parameters only and remains fixed for any number of received pulses. We could not compare our method to this type of recovery since the Intel Core i5 PC with 12GB of RAM used in Section~\ref{sec:SimulationResults} could not store the dictionary.

This drawback is not only a problem of memory, but also of computations. We now analyze the computational complexity of Algorithm~\ref{alg:DopplerFocusingWithGrid}, and compare it to \cite{HighResRadar}. In terms of number of samples, Doppler focusing requires $|\kappa|P$ samples while \cite{HighResRadar} samples at the Nyquist rate and therefore requires $B_h \tau P$ samples. The Doppler focusing operation performs $L |\kappa|$ FFTs of length $M$, and then solves $M$ CS delay recovery problems with dictionary size $|\kappa| \times B_h \tau$. If we assume the complexity of a CS problem is proportional to its dictionary size, then the total complexity of Doppler focusing is $O(L |\kappa| M \log(M) + M |\kappa| B_h \tau)$. In \cite{HighResRadar} a single CS delay-Doppler recovery problem is solved, but with dictionary size $(P B_h \tau) \times (P B_h^2 \tau^2)$, so the recovery's complexity is proportional to $B_h^3$. Thus, even if the dictionary could be stored, processing would scale poorly with increasing bandwidth, which is a critical parameter for high-resolution radar.


\subsection{Two-Stage Recovery}

To overcome the problematic scaling of the simultaneous recovery dictionary, two-stage recovery techniques \cite{KfirWaheed, SubNyquistRadar} separate delay and Doppler estimation, performing them sequentially rather than in parallel. A common approach to performing two-stage recovery uses the multiple measurement vector (MMV) framework, as performed in \cite{BlindMultiband} in the context of undersampling of sparse wideband signals. MMV recovery jointly processes \eqref{eq:CS_p} for $0 \leq p < P$ by stacking the sampled Fourier coefficient vectors and sparse target delay vectors as $\mathbf{C} = [\mathbf{c}_0 \cdots \mathbf{c}_{P-1}] \in \mathbb{C}^{|\kappa| \times P}$ and $\mathbf{X} = [\mathbf{x}_0 \cdots \mathbf{x}_{P-1}] \in \mathbb{C}^{N_\tau \times P}$ accordingly, obtaining
\begin{equation}
\mathbf{C} = \frac{1}{\tau} \mathbf{H} \mathbf{V} \mathbf{X},
\end{equation}
where the dictionary $\mathbf{A} = \frac{1}{\tau} \mathbf{H} \mathbf{V}$ is as in \eqref{eq:CS_single_nu_2L}. MMV recovery algorithms (\eg Simultaneous OMP \cite{StructuredCompressedSensing}) exploit the joint sparsity of $\mathbf{X}$, \ie the fact that the support of $\mathbf{x}_p$ remains constant for all $p$, so $\mathbf{X}$ has at most $L$ nonzero rows. This joint sparsity is used by taking the norm of the rows of $\mathbf{A}^T \mathbf{C}$. This operation can be seen as a form of non-coherent integration from radar literature, since the phase information is destroyed by the norm operator. The norm is a non-linear operation which mixes together signal and noise components, so defining a simple SNR measure as we did in \eqref{eq:SNR1} and \eqref{eq:SNR4} is not possible. Non-coherent integration is a common practice in radar and has been analyzed extensively. Several sources develop approximations of the SNR increase for multiple pulses using non-coherent integration: \cite{Minkoff} estimates it at $\sqrt{P}$ for $P>4$, while \cite{Richards, NonCoherentGain} estimates $P^\beta$ where $0.5 < \beta < 0.833$, with $\beta$ decreasing towards 0.5 as $P$ increases. Regardless of the exact value of $\beta$, Doppler focusing, which compensates for the exact phase differences in the signal, generates an SNR increase linear with $P$, better than two-stage recovery. We compare our method to this type of recovery in Section~\ref{sec:SimulationResults}, and show that the difference in SNR is substantial as predicted by the theory.

\section{Simulation Results}
\label{sec:SimulationResults}

We now discuss how the user defined performance metric influences grid size and Fourier coefficient selection, and then show numerical examples comparing our method to other recovery techniques.

\subsection{Performance Metric}

Our problem lies in a continuous, analog world. When we choose to solve it using CS, which is an approach developed for discrete problems, we must discretize the delay grid, denoting grid step as $\Delta_\tau$. As real world targets delays do not lie on any predefined grid, but our CS recovery assumes they do, it seems we should take $\Delta_\tau \rightarrow 0$ in order to minimize quantization errors. Computational requirements aside, there is a significant drawback to such a decrease in grid step - columns in the CS dictionary $\mathbf{A}$ become ever more similar, making it increasingly coherent, where coherence is defined as the largest absolute inner product between any two columns $\mathbf{a_i}$, $\mathbf{a_j}$ of $\mathbf{A}$:
\begin{equation}
\mu(\mathbf{A}) = \underset{i \neq j}{\max} \frac{|\langle \mathbf{a_i}, \mathbf{a_j} \rangle |}{\| \mathbf{a_i} \|_2 \| \mathbf{a_j} \|_2 }.
\end{equation}
A basic premise of CS ties low coherence to successful recovery \cite{StructuredCompressedSensing}. Therefore, $\mathbf{A}$ is usually designed to have small coherence. This contradicts taking the step size to be increasingly small.

Here, we relinquish this basic assumption, and argue that depending on the chosen performance metric, high coherence can actually help recovery instead of harm it. For example, assume we are interested in delay recovery but are tolerant of some small error $\tau_{max}$. In radar applications, a common performance metric is the ``hit-or-miss" criterion on the estimated delays $\{ {\hat{\tau}_\ell} \}_{\ell=0}^{L-1}$:
\begin{align}
\label{eq:hit_or_miss}
\nonumber e_\ell(\hat{\tau}_\ell) &= \left\{
  \begin{array}{l l}
    0, & \quad \text{if $|\tau_\ell - \hat{\tau}_\ell| < \tau_{max}$}\\
    1, & \quad \text{otherwise}\\
  \end{array} \right. \\
e &= \sum_{\ell=0}^{L-1} e_\ell.
\end{align}
Translating $\tau_{max}$ to a condition on support recovery, \eqref{eq:hit_or_miss} tolerates an error of no more than $K = \lfloor \tau_{max} / \Delta_\tau \rfloor \in \mathbb{N}$ places in the recovered indices from \eqref{eq:CS_single_nu_2L}. Instead of designing $\mathbf{A}$ so that each column is as non-correlative with the other columns as possible, we should design the dictionary so that each column is correlative with its $K$ nearest neighbours, and only afterwards does the correlation drop. This will improve recovery performance in noisy scenarios since in cases where the correct column is not recovered, any one of its similar neighbouring columns still has a chance to overcome the noise and produce a ``hit". Graphs in the next subsection will show how such a coherent $\mathbf{A}$ actually improves recovery performance compared with a less coherent dictionary in very noisy scenarios.

If $\tau_{max}$ is chosen so that $K=0$, the same line of thought entails requiring each dictionary column to be similar to 0 neighbours, \ie for $\mathbf{A}$ to have minimal coherence. Since most CS works deal with exact recovery, this explains why they strive for a minimal $\mu(\mathbf{A})$.

We can control the level of $\mathbf{A}$'s coherence by choosing different sets of Fourier coefficient $\kappa$ in \eqref{eq:c_kappa}. Fig.~\ref{fig:coherence} shows an example of the column correlation pattern for two sets of Fourier coefficients: a consecutive set and a random set, where all coefficients were chosen in $[-B_h/2,B_h/2]$. We define the column correlation function for some column $\mathbf{a_i}$ as
\begin{equation}
\mu_i[j] = \frac{|\langle \mathbf{a_i}, \mathbf{a_j} \rangle |}{\| \mathbf{a_i} \|_2 \| \mathbf{a_j} \|_2 }.
\end{equation}
The consecutive set is better suited for performance criteria which allow some error in support recovery, while the random set will achieve better performance when exact recovery is required.

\begin{figure}[ht]
  \centering
    \includegraphics[width=0.7 \columnwidth]{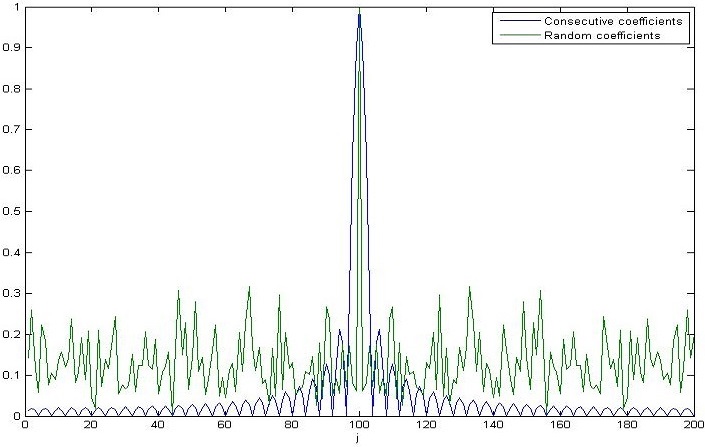}
  \caption{Column correlation pattern $\mu_0[j]$ of the CS dictionary from \eqref{eq:CS_single_nu_2L} for two sets $\kappa$. The consecutive set achieves coherence of 0.9 due to the many correlative columns in the center, while the random set has coherence of 0.3.}
  \label{fig:coherence}
\end{figure}


\subsection{Numerical Results}

We now present some numerical experiments illustrating the recovery performance of a sparse target scene. We corrupt the received signal $x(t)$ with an additive white Gaussian noise $n(t)$ with power spectral density $S_n(f)=N_0/2$, bandlimited to $x(t)$'s bandwidth $B_h$. We define the signal to noise power ratio for target $\ell$ as
\begin{equation}
\label{eq:SNR_l}
\mbox{SNR}_\ell=\frac{\frac{1}{T_p}\int_0^{T_p}|\alpha_\ell h(t)|^2 dt}{N_0 B_h},
\end{equation}
where $T_p$ is the pulse time. The scenario parameters used were number of targets $L$=5, number of pulses $P$=100, PRI $\tau$=10$\mu$sec, and $B_h$=200MHz. Target delays and Doppler frequencies are spread uniformly at random in the appropriate unambiguous regions, and target amplitudes were chosen with constant absolute value and random phase. The classic time and frequency resolutions (``Nyquist bins"), defined as $1/B_h$ and $1/P\tau$, are 5$n$sec and 1 KHz accordingly. In order to demonstrate a 1:10 sampling rate reduction, our sub-Nyquist Xampling scheme generated 200 Fourier coefficients per pulse, as opposed to the 2000 Nyquist rate samples. We tested Doppler focusing with two types of Fourier coefficient sets $\kappa$, a consecutive set and a random set. We compared Doppler focusing recovery performance with classic processing and a two-stage recovery method as described in \cite{KfirWaheed} (where we use a CS algorithm instead of ESPRIT) using the following criteria:

\begin{enumerate}
\item \textbf{Hit-Rate} -- we define a ``hit" as a delay-Doppler estimate which is circumscribed by an ellipse around the true target position in the time-frequency plane. We used an ellipse with axes equivalent to $\pm3$ times the time and frequency Nyquist bins.
\item \textbf{Recovery RMS error} -- for ``hits", we measure the root mean square error in both time and frequency.
\end{enumerate}

As noted in the previous section, a single stage CS recovery method using Nyquist bins spacing consumes a prohibitive amount of memory and was not able to run on any computer at hand, since the CS dictionary required storing $4 \cdot 10^9$ elements (occupying $32GB$ of memory using standard IEEE double precision): $\frac{2\pi/\tau}{2\pi/P\tau} \frac{\tau}{1/B_h} = P\tau B_h = 2 \cdot 10^5$ columns and $P \tau B_h / 10 = 2 \cdot 10^4$ measurements per column.

\begin{figure}[ht]
  \centering
    \includegraphics[width=0.8 \columnwidth]{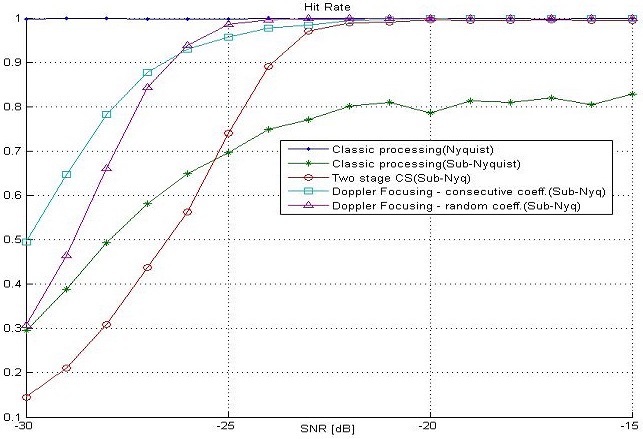}
  \caption{Hit Rate for classic processing, two-stage CS recovery and Doppler focusing. Sub-Nyquist sampling rate was one tenth the Nyquist rate.}
  \label{fig:HitRate_SNR}
\end{figure}

\begin{figure}
  \centering
    \includegraphics[width=0.8 \columnwidth]{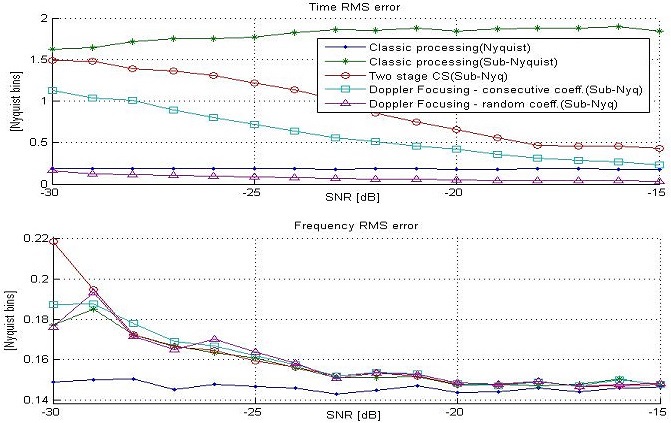}
  \caption{RMS error of time and frequency estimates for classic processing, two-stage CS recovery and Doppler focusing. Sub-Nyquist sampling rate was one tenth the Nyquist rate.}
  \label{fig:RMS_SNR}
\end{figure}

For CS-based techniques, the delay grid step $\Delta_\tau$ was chosen as half a Nyquist bin. For Doppler focusing, the Doppler frequency region was discretized with uniform steps of half a Nyquist bin. To provide a fair comparison, classic processing was performed using identical bin sizes.

Fig.~\ref{fig:HitRate_SNR} and Fig.~\ref{fig:RMS_SNR} demonstrate the hit-rate and RMS error performance of the different recovery methods for various SNR values. It is evident that Doppler focusing is superior to the other sub-Nyquist recovery techniques. Between the two Doppler focusing approaches, consecutive coefficients are better suited for lower SNR, while choosing coefficients randomly improves performance as SNR increases. Since both sets $\kappa$ produce CS dictionaries with column correlation functions which are not matched to the ``hit-or-miss" performance criteria used, we have no reason to assume one should be better than the other. Also, random coefficients, when producing a hit, have very small delay errors (even compared with Nyquist rate classic processing) due to low CS dictionary coherence. As opposed to Doppler focusing recovery performance which decreases gracefully with sample rate reduction, classic processing suffers significantly when sample rate is reduced below Nyquist.

Fig.~\ref{fig:HitRate_SNR_2} shows the same hit rate graph for classic processing, but this time the waveform used for Doppler focusing had its CTFT adjusted so that energy was transferred from frequencies which were not sampled, to those that were. This was performed by passing the signal through a low pass filter and rescaling its amplitude so target SNR \eqref{eq:SNR_l} remains constant. Since Doppler focusing imposes no restrictions on the transmitter, we can use a signal with the same total energy, but have it spread out in a manner which is more favorable to the frequency domain sampling used in Xampling. Since performance for Doppler focusing improves significantly, we are able to obtain excellent recovery results at much lower SNR values, surpassing classic processing which uses ten times as many samples. This shows that the performance degradation caused by a sub-Nyquist sampling rate can be compensated for using a suitable transmitter. The drawback of using such a narrowband signal can be seen in Fig.~\ref{fig:narrowband}, where we examine resolution in terms of the ability to separate two closely spaced identical targets with equal Doppler frequencies. We see that for very close targets, classic processing using a wideband signal is able to distinguish the two targets far better than Doppler focusing recovery using a narrowband signal.

\begin{figure}[ht]
  \centering
    \includegraphics[width=0.8 \columnwidth]{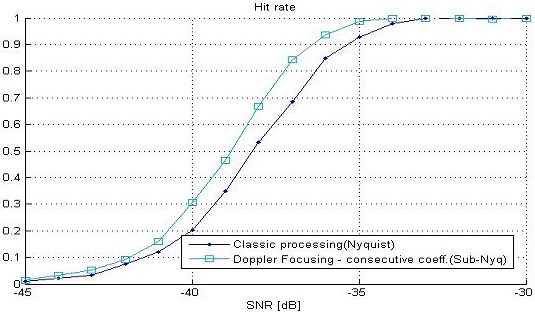}
  \caption{Hit Rate for classic processing and Doppler focusing at one tenth the Nyquist rate, where the waveform used for Doppler focusing had its entire energy contents concentrated in the sampled frequencies.}
  \label{fig:HitRate_SNR_2}
\end{figure}

\begin{figure}[ht]
  \centering
    \includegraphics[width=0.9 \columnwidth]{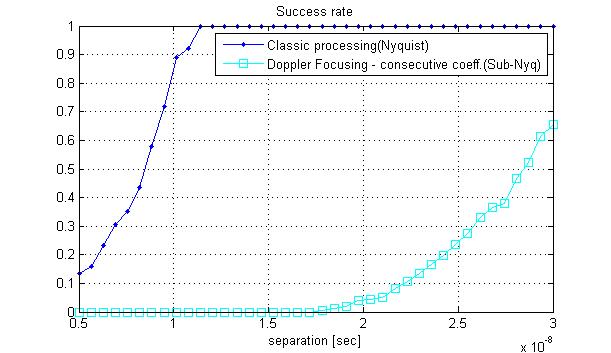}
  \caption{Probability to distinguish two separate closely spaced identical targets, where classic processing uses a tenfold wideband signal compared with Doppler focusing.}
  \label{fig:narrowband}
\end{figure}

Fig.~\ref{fig:map} shows the sparse target scene on a time-frequency map for a -28dB SNR scenario, where each target is displayed along with its hit rate ellipse, together with the various sub-Nyquist recovery methods' estimates and hit rates. As noted in Section~\ref{sec:DopplerFocusing}, only Doppler focusing is able to distinguish between the two targets having almost identical delays (around 4.2 $\mu$sec) but different Doppler frequencies.

\begin{figure}[ht]
  \centering
    \includegraphics[width=0.8 \columnwidth]{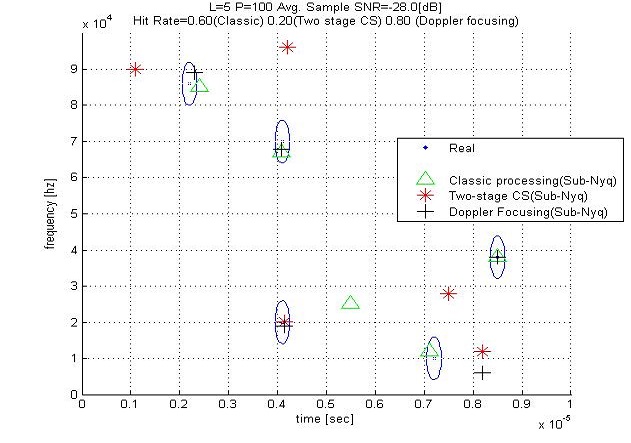}
  \caption{Real target positions along with various estimates. Doppler focusing achieves highest hit rate among sub-Nyquist methods. Only Doppler focusing detects two targets around 4.2$\mu$sec. (Hit rate ellipses were enlarged for clarity.)}
  \label{fig:map}
\end{figure}

Regarding target dynamic range, Fig.~\ref{fig:dynamic_range} demonstrates the advantage of using Doppler focusing to resolve closely spaced targets with different powers. In this scenario, two targets satisfying $20\log_{10}(|\alpha_1| / |\alpha_2|) = 20 [dB]$ are placed adjacently, so that their hit rate ellipses intersect. Doppler focusing based recovery at one tenth the Nyquist rate generates two hits, while MF processing at Nyquist rate recovers only one of the targets.

\begin{figure}[ht]
\centering
    \includegraphics[width=0.8 \columnwidth]{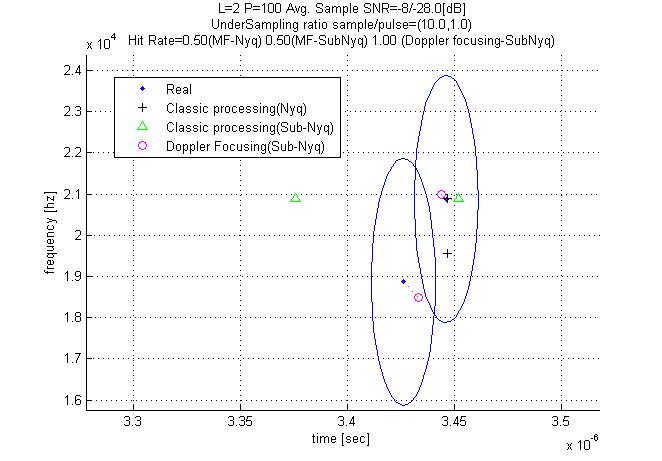}
\caption{Target scene composed of two closely spaced targets, the target on the left 20dB more powerful than the target on the right. MF processing at both Nyquist and one tenth the Nyquist rate recovers only one of the two targets, while Doppler focusing recovers both. No windowing was used ($w[p]=1$).}
\label{fig:dynamic_range}
\end{figure}

\subsection{Clutter}

We demonstrate the robustness of Doppler focusing to clutter with the following scenario. We simulate nine targets with Doppler frequencies spread uniformly in $[-\pi/\tau, \pi/\tau]$. Clutter is modeled as 4000 Swerling-0 scatterers, with Doppler frequencies distributed uniformly in a single Nyquist bin around DC, to allow for small relative velocity, \eg waves in the sea. Clutter scatterers are distributed uniformly at random at all delays in $[0,\tau]$. The SCR is -50dB, as can be seen in Fig.~\ref{fig:clutter_coeff}, and the targets' SNR \eqref{eq:SNR_l} is -25dB. All other system parameters are as described in the previous subsection.

\begin{figure}[ht]
\centering
    \includegraphics[width=0.7 \columnwidth]{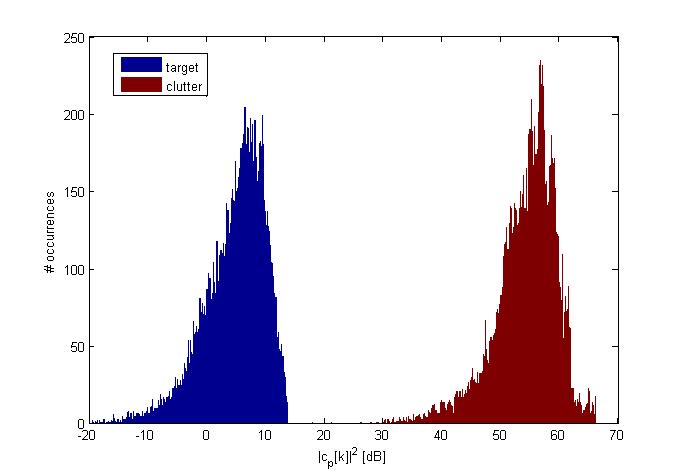}
\caption{Histogram of target vs. clutter Xamples $c_p[k]$. Average SCR for a single Xample is -50dB.}
\label{fig:clutter_coeff}
\end{figure}

We perform detection using Algorithm~\ref{alg:DopplerFocusingWithGrid}, ignoring the clutter ridden DC Doppler bin, and due to leakage effects, also its two nearest neighbours on each side. Since for this very low SCR, clutter sidelobes still cover the targets, we use a Taylor window with -50dB attenuation to improve Doppler frequency isolation. The recovered target scene with and without windowing is shown in Fig.~\ref{fig:clutter_detection}. This simulation shows that any interference which is localized in Doppler, can be manipulated using \eqref{eq:Windowing}, so it has a negligible effect on targets in other Doppler frequencies.

\begin{figure}[ht]
\centering
    \includegraphics[width=0.8 \columnwidth]{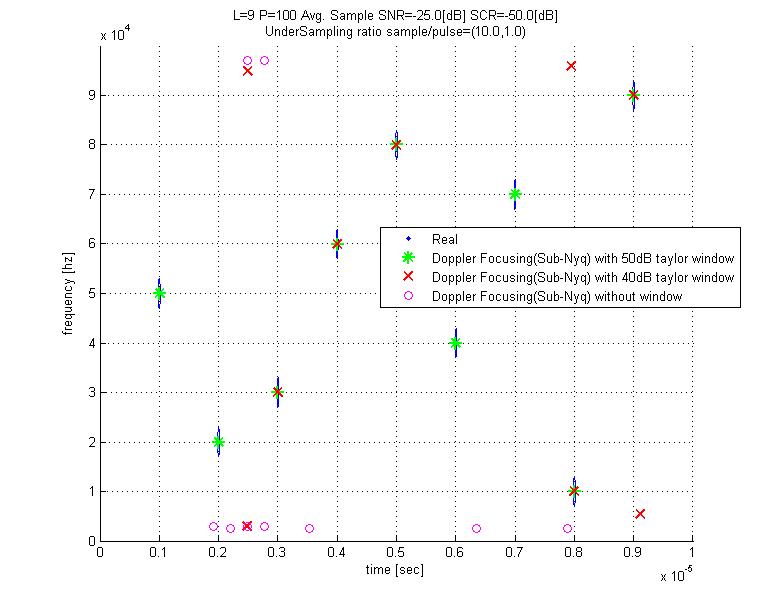}
\caption{Recovered target scene with nine targets and almost static clutter. Without windowing clutter sidelobes permeate the nonzero Doppler frequency area and cause misdetections. With 40dB windowing five out of nine targets are ``hits". With 50dB windowing the entire scene is detected correctly.}
\label{fig:clutter_detection}
\end{figure}

\section{Radar Experiment}
\label{sec:experimental}

In this section we present a real experiment of our radar receiver hardware prototype. Our setup includes a custom made sub-Nyquist radar receiver board which implements sub-Nyquist Xampling and digital recovery using Doppler focusing, while the analog input signal \eqref{eq:x_received} was synthesised using National Instruments hardware. The RF front end and board we use are identical to the ones in \cite{EliGal}, but the digital recovery method accounting for target Doppler frequencies is different. Additional information regarding the system's hardware and configuration issues can be found in \cite{EliGal}.

The experimental process consists of the following steps. We begin by using the AWR software, which enables us to examine a large variety of scenarios, comprised of different target parameters, \ie delays, Doppler frequencies and amplitudes. With the AWR software we simulate the complete radar scenario, including the pulse transmission and accurate power loss due to wave propagation in a realistic medium. The AWR also contains a model of a realistic RF receiver, which simulates the demodulation of the RF signal to IF frequencies, and saves the output to a file. The simulation result is loaded to the AWG module, which produces an analog signal, that enters the radar receiver board and is processed and sampled as described in \cite{EliGal}. The samples are fed into the chassis’ controller and a MATLAB function is launched that runs the Doppler focusing reconstruction algorithm. Our system contains a fully detailed interface implemented in the LabView environment, which allows simple activation of the process. Various target scenes, with different target delays, Doppler frequencies and amplitudes, are recovered successfully using this setup. A screenshot of the interface is depicted in Fig.~\ref{fig:scenes}.

This experimental prototype proves that the sub-Nyquist methodology described in this paper is actually feasible in practice. The recovery method proposed here not only describes digital recovery, but also addresses the problem of sampling the analog signal at a low rate, in a way which is feasible with standard RF hardware.

\begin{figure}[ht]
\centering
    \includegraphics[width=1.0 \columnwidth]{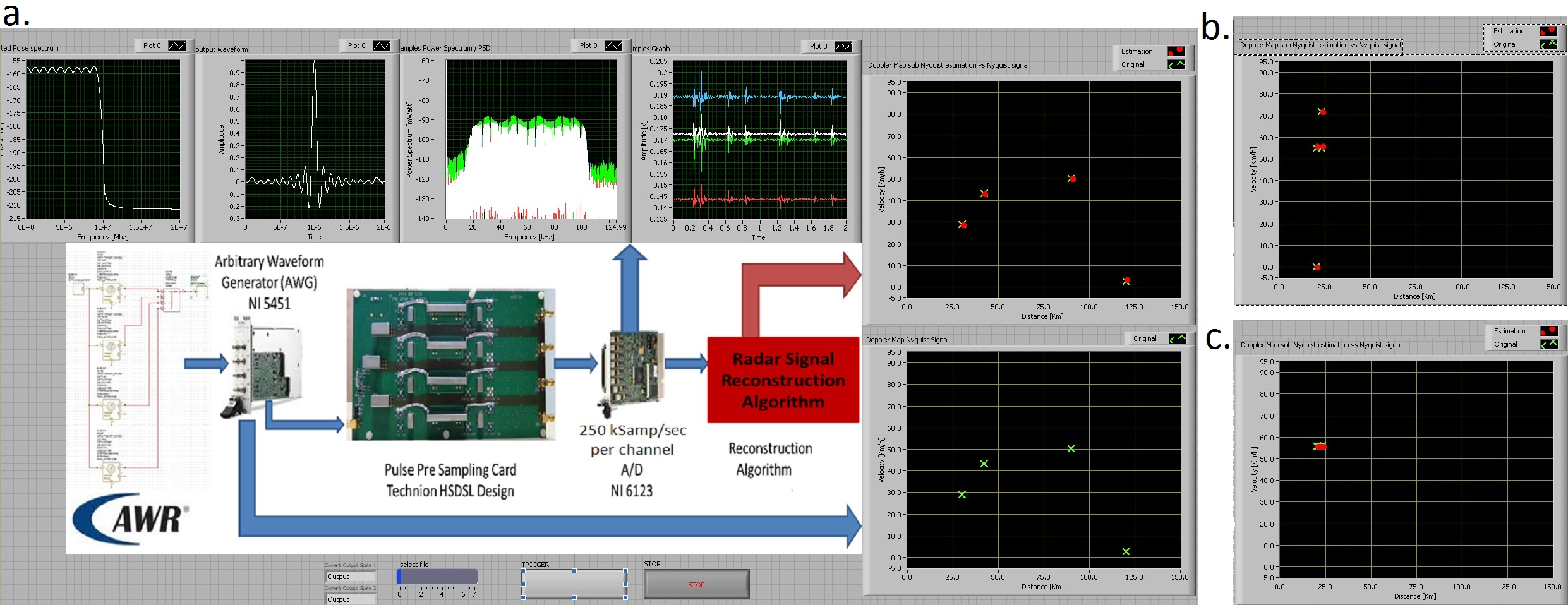}
\caption{The LabView experimental interface. \textbf{a.} From left to right: $H(\omega)$, $h(t)$, the frequency response of each channel, the 4 signals detected in each channel, at the top- the reconstructed target scene, at the bottom- the original target scene. On the right two additional target scenes, with Doppler focusing based recovery successful in both cases: \textbf{b.} All four targets having very closely spaced delays, with two of them also having close Doppler frequencies. \textbf{c.} All four targets have very similar delays and Doppler frequencies.}
\label{fig:scenes}
\end{figure}

\section{Conclusion}
\label{sec:Conclusion}

We demonstrated a radar sampling and recovery method called Doppler focusing, which employs the techniques of Xampling and CS, and is independent of the radar signal's bandwidth. Doppler focusing allows for low rate sampling and digital processing, and imposes no constraints on transmitted signal. It also leads to CS recovery with dictionary size scaling with delay grid size only, and provides SNR scaling which is linear in the number of received pulses, identical to an optimal MF. We compared our method to other sub-Nyquist recovery techniques and have seen its clear advantage in simulations. When sampling at one tenth the Nyquist rate, and for SNR above -25dB, Doppler focusing achieves results almost equal to classic recovery working at the Nyquist rate.

We are currently working on enhancing Doppler focusing to handle the case of an unknown number of targets, and increasing the algorithm's dynamic range by improving treatment of strong targets and the sidelobes they introduce.



\section*{Acknowledgment}
The authors would like to thank the anonymous reviewers for their constructive remarks, which brought to our attention several important issues and helped improve the paper.

\bibliographystyle{IEEEtran}
\bibliography{IEEEabrv,my_references}

\end{document}